\documentclass{statsoc}

\pdfoutput=1

\usepackage[a4paper]{geometry}
\usepackage{times}
\usepackage{comment}

\usepackage{etoolbox}

\makeatletter
\patchcmd{\@makecaption}
  {\parbox}
  {\advance\@tempdima-\fontdimen2} % decrease the width!
  {}{}
\makeatother

\usepackage{natbib}
\usepackage{amsfonts} %For mathbb
\usepackage{amscd} % Diagrams
\usepackage{amsmath} %Align environment
\usepackage{graphicx}
\usepackage{epstopdf}
\usepackage{booktabs}
\usepackage{pgf}
\usepackage{array, multirow}
\usepackage{nicefrac}

\usepackage{hyperref}
\usepackage{amsmath,hyperref,amsfonts,amssymb,mathtools}
\usepackage{bm}
\newcommand{\mb}[1]{\bm{#1}}
\newcommand\ceil[1]{\left\lceil#1\right\rceil}

\newcommand{\<}{\langle}
\renewcommand{\>}{\rangle}

\DeclareMathOperator{\pP}{Pr} % Probability
\newcommand{\pN}{\textrm{N}} % Normal
\newcommand{\pExp}{\textrm{Exp}} % Exponential

\DeclareMathOperator{\E}{E} % Expectation
\DeclareMathOperator{\Var}{var}
\DeclareMathOperator{\Cov}{cov}

\newcommand{\pE}{\E}

\newcommand{\pCov}{\Cov}

\newcommand{\e}{\mathrm{e}}
\newcommand{\oms}{\mb{\omega}_s}
\newcommand{\omt}{\omega_t}

\newcommand{\mv}[1]{\bm{#1}}
\DeclareMathOperator{\diag}{diag}
\newcommand{\mat}[1]{\begin{bmatrix}#1\end{bmatrix}}

\newcommand{\wt}[1]{\widetilde{#1}}

\newcommand{\wh}[1]{\widehat{#1}}

\newcommand{\mmd}{\mathrm{d}}
\newcommand{\md}{\,\mathrm{d}}

\newcommand{\s}{\bm{s}}

\newcommand{\W}{\mathcal{W}}
\newcommand{\EE}[1][]{\mathcal{E}_{#1}}
\newcommand{\dEE}[1][]{\md\EE[#1]}

\newcommand{\RR}{\mathbb{R}}
\newcommand{\sphere}{\mathbb{S}}
\newcommand{\cD}{\mathcal{D}}
\newcommand{\ZZ}{\mathbb{Z}}

\newcommand{\Mo}{{\text{M}_\text{0}}}
\newcommand{\Ma}{{\text{M}_\text{A}}}
\newcommand{\Mb}{{\text{M}_\text{B}}}
\newcommand{\Mc}{{\text{M}_\text{C}}}
\newcommand{\Md}{{\text{M}_\text{D}}}

\newcommand{\qqed}{\qed}
\newtheorem{lemma}{Lemma}
\newtheorem{proposition}{Proposition}
\newtheorem{theorem}{Theorem}
\newtheorem{corollary}{Corollary}

\newtheorem{definition}{Definition}

\graphicspath{{figures/}}

\title[A diffusion-based spatio-temporal extension of Gaussian Mat\' ern fields]{A diffusion-based spatio-temporal extension of Gaussian Mat\' ern fields}
\author{Finn Lindgren\footnote{{\it Address for correspondence:} Finn Lindgren, School of Mathematics, The University of Edinburgh, James Clerk Maxwell Building, Peter Guthrie Tait Road, Edinburgh EH9 3FD, Scotland, UK\newline E-mail: Finn.Lindgren@ed.ac.uk}}
\address{
School of Mathematics, The University of Edinburgh, Scotland}
\author{Haakon Bakka}
\address{Norwegian Veterinary Institute, {\AA}s, Norway}
\author[Finn Lindgren, Haakon Bakka, David Bolin, Elias Krainski and H{\aa}vard Rue]{David Bolin, Elias Krainski  and H{\aa}vard Rue}
\address{CEMSE Division, King Abdullah University of Science and Technology, Saudi Arabia}

\begin{document}

\begin{abstract}
        Gaussian random fields with Mat\'ern covariance functions are popular models in spatial statistics and machine learning.
        In this work, we develop a spatio-temporal extension of the Gaussian Mat\'ern
        fields formulated as solutions to a stochastic partial differential equation.
        The spatially stationary subset of the models have marginal spatial Mat\'ern covariances,
        and the model also extends to Whittle-Mat\'ern fields on curved manifolds, and to more general
        non-stationary fields. In addition to the parameters of the spatial dependence
        (variance, smoothness, and practical correlation range) it additionally has parameters
        controlling the practical correlation range in time, the smoothness in time, and
        the type of non-separability of the spatio-temporal covariance. Through the
        separability parameter, the model also allows for separable covariance functions.
        We provide a sparse representation based on a finite element approximation, that
        is well suited for statistical inference and which is implemented in the R-INLA software.
The flexibility of the model is illustrated in an application to spatio-temporal modeling of global temperature data.
\end{abstract}

\section{Introduction}

\subsection{Modelling spatio-temporal data}
Statistical models for spatio-temporal data have
applications in areas ranging from the analysis of
environmental data \citep{cameletti2013}
and climate data \citep{wood2004climate,fuglstad2020compression},
to resource and risk modeling (e.g., of wildfires, \citet{serra2014wildfires}),
 disease modeling \citep{inlaexample2015malaria,moraga2019geospatial},
and ecology \citep{yuan2017,zuur2017spatial}.
These models typically use spatio-temporal random effects,
defined as Gaussian spatio-temporal stochastic processes
and rely on a large body of theoretical and methodological literature
\citep[][and references therein]{stein2012interpolation, gelfand2010handbook, CressieWikle2011}.

At best, this theory is carefully studied when the spatio-temporal model is constructed,
so that the model with the most appropriate assumptions can be used.
In practice, however, users of statistical software often choose
a model based on convenience.
If there are available code examples,
the choices made in these will often be carried forward into future analyses.
For example, users of R-INLA \citep{rue2009,rue2017bayesian,art664,art693,art702,art703} construct
space-time models through
Kronecker products of a spatial Mat\' ern model,
and first- or second-order autoregressive models in time,
following the code examples in  \citet{krainski2018advanced}.
This paper is aimed at improving the general practice of space-time data analysis,
by providing a new family of spatio-temporal stochastic processes
for use as random effects in statistical software.

We will mainly discuss
stochastic processes $u(\s,t)$
that are stationary and spatially isotropic, i.e., the
covariance function can be written as
$\pCov(u(\s_1, t_1), u(\s_2, t_2)) = R(h_s, h_t),$
where $h_s = ||\s_1 - \s_2||$ and $h_t = |t_1 - t_2|$,
but will also extend these process models to spatial non-stationarity
and processes on general manifolds.
We consider these stochastic processes in the context of hierarchical models,
as a latent model component, observed through some measurement process,
with no direct measurements of the stochastic process itself.
Consider, for example, a model with a linear predictor
\begin{align}
\eta(\s, t) = \sum_{i=1}^m X_i(\s,t) \beta_i + f_1\{z_1(\s,t)\} +
  \ldots + f_k\{z_k(\bm s,t)\} + u(\s, t),
  \label{eq:pred}
\end{align}
that is connected to the response $y$ through
some likelihood
or loss function \citep{bissiri2016general} such that $\pE\{y(\s,t)\} = g\{\eta(\s,t)\}$ for some fixed and known function $g$.
Here $X_i$ and $z_j$ are covariates that vary over both space and time, $\beta_i$ the regression coefficient for the fixed effects, and
$f_j(z_j)$ are random effects. Typical examples are splines and latent Gaussian processes
used to approximate the effect  of altitude or distance to coastline.
%Further, $u(\mb s, t)$ is the spatio-temporal stochastic process we discuss in this paper,
%which we will refer to as the spatio-temporal model component whenever we are considering it as part of a hierarchical model.
%The $f_j$ terms can also include a temporal trend, a spatial reference level,
%and time-varying or space-varying regressions;
%possibly partially confounded with the spatio-temporal model component.
%
This common situation with a stochastic process as a model component
impacts the methodological considerations we make.
The predictor is also a spatio-temporal stochastic process, with
a covariance function that can be deduced from
the assumptions on the
model components. However, properties of the predictor that we may
discover by investigating the covariance function of the predictor may not be shared by the spatio-temporal model component $u$ because of the other factors.
Hence, we may have little prior information about the covariance structure of the
spatio-temporal model component,
except that
it should be physically realistic,
and should mimic the dependency structure in models
of physical processes.

Users of software for spatio-temporal modelling
most often use separable models (see, e.g., \citet{bakka2018spatial,krainski2018advanced}), i.e., models where $u$ has a covariance function of the form $R( h_s, h_t) = R_s(h_s) R_t(h_t),$
for some spatial and temporal marginal covariance functions $R_s(\cdot)$ and $R_t(\cdot)$.
This is typically not because this is a desired property, but since such models are
readily available in statistical software, and there are many good arguments
for why models should not be assumed separable, see
\citet{stein2005space}, \citet{cressie1999classes},
\citet{fonseca2011general}, \citet{rodrigues2010class},
\citet{gneiting2002nonseparable}, \citet{sigrist2015stochastic},
\citet{wikle2015modern}.

\subsection{The Mat\'ern family of covariance functions}
The most well known family of covariance functions for stationary random fields on $\mathbb{R}^d$ is
the Mat\'ern covariance,
\begin{align}
R_M(h) =   \frac{\sigma^2}{2^{\nu-1} \Gamma(\nu)} \left( \kappa h\right)^\nu K_\nu \left( \kappa h \right),
\end{align}
where $\nu,\kappa>0$ are smoothness and scale parameters, $\sigma^2$ is the variance of the corresponding random field,
$K_\nu $ is the bessel function of the  second kind of order $\nu$, and
$\Gamma$ is the Gamma function.
An important property of this covariance faily is that it allows for explicit control of the differentiability of the corresponding stochastic process through the parameter $\nu$. It further allows for control of the practical correlation range $r = \sqrt{8\nu}/\kappa$ \citep{lindgren2011explicit}.
The covariance function is usually attributed
to \citet{matern1960spatial}, and it was advocated early by
\citet{handcock1993bayesian} and \citet{stein2012interpolation}.
See \citet{guttorp2006studies} for  a historical account of the covariance function and its
connections to various areas in physics.

The goal of this paper is to extend the Mat\' ern covariance function to
a family of spatio-temporal covariance functions. One way of doing this would be to
extend the covariance function to a spatio-temporal covariance.
However, we argue that it is better to base the extension on some of the other equivalent mathematical representations,
or \emph{views}, of Gaussian Mat\' ern fields.
One such alternative representation is the stochastic partial differential equation (SPDE) representation by \citet{whittle1963stochastic}.
Specifically, a Gaussian Mat\'ern field on $\mathbb{R}^d$ solves the SPDE
\begin{align}
(\kappa^2 - \Delta)^{\alpha/2} u = \mathcal W, \label{eq:simpleSpde}
\end{align}
where $\kappa>0$, $\Delta$ is the Laplacian,
 $\mathcal W$ is Gaussian white noise, and $\alpha = \nu + d/2$.
Via the SPDE representation, we note that a Gaussian Mat\'ern field has
precision operator ${Q= (\kappa^2 - \Delta)^{\alpha}}$.
The precision operator (as well as the pseudo-differential operator $(\kappa^2 - \Delta)^{\alpha/2}$) are defined in terms of Fourier transforms \citep{lindgren2011explicit}, and informally, we get the Fourier transform
of the precision operator by replacing derivatives with $d$-dimensional wave-numbers $\bm w$.  For any precision operator which is a polynomial in the Laplacian,
$Q = p(-\Delta)$, such as the Mat\'ern operator with $\alpha \in \mathbb N$,  this results in a polynomial
$\mathcal F(Q) = p(\|\bm w\|^2)$.
This function is the reciprocal of the spectrum of the Gaussian process,
illustrating why many common spectrums are the reciprocal of an even polynomial.
In fact, \citet{rozanov1977} showed that a stationary stochastic process on $\mathbb{R}^d$ is Markov if and only if
the spectral density is the reciprocal of a polynomial,
and more generally, a stochastic process is Markov if the precision operator is a local operator, which is the case for integer powers of the Laplacian.
For further details on the theory of the SPDE representation, see \citet{kelbert2005fractional, prevot2007concise, lindgren2011explicit, bolin_rational_2020}.

We could also represent a Gaussian Mat\'ern field as a stochastic integral with respect to white noise. For Gaussian Mat\'ern fields, the kernel in the integral representation is the Green's function of the differential operator \citep[see, e.g.,][]{bolin13}.
This representation can be used to define other valid covariance functions by replacing the Green's function with some other kernel
 \citep[see, e.g.,][]{fuentes2002spectral, higdon2002space, rodrigues2010class}.

The modeling approaches stemming from these different views of the Gaussian Mat\'ern fields can be thought of as implicit and explicit. In implicit approaches such as the covariance-based representation,
one does not have a direct formulation of the process itself, and properties of interest need to
be derived from the covariance function.
In explicit, or constructive, approaches
one directly defines the process through, e.g., an SPDE or a stochastic integral
with the desired properties encoded.
In this paper we follow the explicit approach to construct
a stochastic process based on diffusion processes.
Other properties, such as covariance non-separability,
are then merely consequences of the explicit construction.

\subsection{SPDE-based spatio-temporal generalisations of the Mat\' ern covariance family}

There is a large literature on spatio-temporal covariance models \citep[see, e.g.,][and the references within]{Porcu2021}. Broadly, models for spatio-temporal  Gaussian random fields can be divided into two categories; the implicit second-order covariance based models and explicit dynamical models \citep{CressieWikle2011, roques_spatial_2022}.
It should be noted that \citet{Porcu2021}, contrary to this terminology, classifies the SPDE-based methods as implicit since they do not explicitly specify the covariance function. However, the covariance functions is merely a property
of the process, and only indirectly defines the process family, whereas dynamical models
directly determine the spatial and temporal evolution of the process.  As shown by
\citet{lindgren2011explicit}, the covariance does not have an inherent advantage
over spectral and precision operator/matrix methods, for practial applications and computations.

The second-order model specifications specify the Gaussian process properties
by specifying its first two moments, and are thus based on formulating valid spatio-temporal covariance functions. In dynamical model specifications, the evolution of the Gaussian process is explicitly described either by specifying the conditional distributions of the current state of the process given its past through conditional distributions \citep[e.g.,][]{storvik2002stationary}, or by specifying the process as the solution to an SPDE \citep{CressieWikle2011}. One of the advantages with the dynamical approach is that it avoids the difficulties with formulating flexible and yet valid spatio-temporal covariance functions that can possess features such as non-separability or non-stationarity. In this work we focus on dynamical models specified through SPDEs, which makes extension to non-stationary fields
and manifold models straightforward.

Several papers have been using the SPDE view to suggest models for spatio-temporal
stochastic processes. A common extension of Mat\'ern covariance fields to space-time
is to use it as the spatial component in a separable model.
\citet{jones1997models} discuss how separable
covariance functions can be understood through differential operators,
written as $L = L_s L_t$,
where $L_s$ is a purely spatial operator and
$L_t$ is a purely temporal operator.
In agreement with \citet{jones1997models}, we note that
these operators are almost never encountered when modeling physical reality,
hence, separable models are typically not physically motivated models
for the spatio-temporal process.

\citet{whittle1963stochastic} considered a spatio-temporal stochastic process formulated as a solution to
\begin{align}
\frac{\partial u}{\partial t} + (\kappa^2 - \Delta) u(\mv{s},t) = \epsilon(\mv{s},t), \label{eq:whittlespacetime}
\end{align}
where $\epsilon(\bm s,t)$ is a stationary spatio-temporal noise process. \citet{Whittle1986} denoted the model as a ``diffusion-injection model'' since it is a diffusion processes with stochastic variability ``injected'' through the noise process on the right-hand side. Despite being a natural spatio-temporal extension of the Mat\'ern model \eqref{eq:simpleSpde} with $\alpha=2$, the model does not have any flexibility in terms of differentiability in space or time. \citet{jones1997models} proposed a generalization, with greater flexibility for the marginal spatial covariances, by considering the fractional SPDE
\begin{align}\label{eq:jones}
\left( \frac{\partial}{\partial t} + \left( \kappa^2 - \Delta \right)^{\alpha/2} \right) u(\mv{s},t)= \dEE(\mv{s},t),
\end{align}
where $\dEE$ is space-time Gaussian white noise. When requiring spatial operator order $\alpha>d$, this SPDE has regular continous solutions. In order to allow smaller operator orders $\alpha$, such as a dampened ordinary diffusion operator with $\alpha=2$ on $\mathbb{R}^2$, as in \citet{whittle1963stochastic}, the driving noise process would need to have spatial dependence. We will make this precise in later sections.
An advantage with \eqref{eq:jones} is that the spatial smoothness can be controlled, since the solutions on the spatial domain $\RR^d$ have smoothness $\nu_s = \alpha - d/2$.
The disadvantage is that the temporal smoothness also is determined by $\alpha$. As we will see later, the marginal temporal differentiability of the solution,  is $\nu_t = (1-d/\alpha)/2$.

A model with general differentiability in both space and time was formulated by \citet{stein2005space}, who consider Gaussian spatio-temporal models specified through the spectrum
\begin{equation}\label{eq:sstein}
S(\bm{w}_s, w_t) = \{
c_1(a_1^2+\|\bm{w}_s\|^2)^{\alpha_1} +
c_2(a_2^2+|w_t|^2)^{\alpha_2} \}^{-\nu},
\end{equation}
where $c_1>0, c_2>0, a_1, a_2$ are scale parameters, $a_1^2 + a_2^2 >0$,
$\alpha_1$, $\alpha_2$ and $\nu$ are
smoothness parameters with further restrictions in order to obtain a model with finite variance. For example, on a two-dimensional spatial and one-dimensional temporal domain,
$2/\alpha_1 + 1/\alpha_2 < 2 \nu$ is required.
Stein's model can also be stated as an SPDE driven by space-time white noise,
\begin{align}\label{eq:steinspde}
\left(c_1(a_1^2 - \Delta)^{\alpha_1} +
c_2\left(a_2^2-\frac{\partial^2}{\partial t^2}\right)^{\alpha_2}\right)^{\nu/2}u(\mv{s},t) = \dEE(\mv{s},t),
\end{align}
see \citet{krainski2018statistical} and \citet{vergara_general_2022}.
A related model based on spectral densities, which also has separable models as a special case, was considered by \citet{fuentes2008class}.

The case $\alpha=2$ of \eqref{eq:jones} for general dimension was considered in
\citep[][Section 3.5]{lindgren2011explicit}, suggesting the generalisation
\begin{align}
\left( \frac{\partial}{\partial t} + \kappa^2 + \mb m \cdot \nabla - \nabla \cdot \mv{H} \nabla \right) u(\mb s, t) = \dEE[Q](\mv{s}, t),
\end{align}
where $\mv{H}$ is a constant diffusion matrix, $\mb m$ is an advection (transport) vector field, and the innovation process $\dEE[Q](\mv{s}, t)$ white noise in time but is sufficiently smooth in space to generate regular solutions $u(\mb s,t)$; see \citet{lindgren2011explicit} and \citet[Sec 2.2]{sigrist2015stochastic}.
Physically, this model might be interpreted as a dampened advection-diffusion process, with the driving mechanism of the space-time field, such as introducing new mass (or, particles) into the system, having positive spatial correlation. See also \citet{liu_statistical_2022,clarottoetal2022advdiff}.

In this work, we introduce another generalisation of the models by \citet{jones1997models} and \citet{lindgren2011explicit}
that intersects, but is otherwise distinct from, the Stein model family.

\subsection{Outline}
In Section \ref{sec:demfdefine} we introduce a new family of
SPDE-based spatio-temporal stochastic processes.
Model properties such as spatial and temporal differentiability, and parameter interpretations, are presented in Section \ref{sec:properties}.% and priors for the model parameters are presented in Section \ref{sec:pcprior}.
We present a sparse basis function representation in
Section \ref{sec:gmrfrepr},
and an implementation in R-INLA \citep{rue2009}
in Supplementary Materials,
which allows us to construct models with different likelihoods
and several random effects in a generalised additive model context.
In Section \ref{sec:applic}, we present a forecasting example
that illustrates clearly the difference between
separable models and non-separable diffusion-based models,
and an application to a global temperature dataset. The article concludes
with a discussion in Section~\ref{sec:discussion}.

% !TEX root = spacetime_jrssb.tex
\section{A diffusion-based family of spatio-temporal stochastic processes}
\label{sec:demfdefine}

In this section we define a
diffusion-based extension of the Gaussian Mat\'ern
fields to a family of spatio-temporal
stochastic processes (abbreviated DEMF).
The main property we aim for is that the
process should be a Gaussian Mat\'ern field when considered for a
fixed time point in $\RR^d$. That is, then the process is considered on the spatial domain $\cD= \mathbb{R}^d$, the spatial marginalisations of the process have Mat\'ern covariances.
When the models are considered on a general (compact) manifold $\cD$, the spatial marginalisations are solutions to a generalised spatial Whittle-Mat\'ern model on $\cD$
\citep{lindgren_spde_2022}.

Consider again the operator $L_s=\gamma_s^2 - \Delta$ on a spatial domain $\cD$,
including any boundary conditions needed for compact domains.
and introduce the precision operator for the generalised Whittle-Mat\'ern covariances as
$Q(\gamma_s, \gamma_e, \alpha)=\gamma_e^2 L_s^{\alpha}$,
corresponding to solutions $v(\mb s)$ to the spatial stochastic SPDE
\begin{align}
\gamma_e L_s^{\alpha/2} v(\mb s) &= \W(\mb s), \qquad \mv{s} \in \cD
\end{align}
where $\mathcal{W}$ is a spatial white noise process, as discussed by
\citet{whittle1963stochastic} and \citet{lindgren2011explicit}.
When $\cD=\RR^d$, and a stationary condition is imposed, these processes are regular
Mat\'ern processes.
We then define a noise process $\dEE[Q](\mb s,t)$ as Gaussian noise
that is white in time but correlated in space, with precision operator $Q=Q(\gamma_s,\gamma_e,\alpha_e)$ for some non-negative $\alpha_e$.
For $a>0$, the cumulative time-integral process
\begin{align}
\EE[Q](\mb s,(0,a]) = \int_{t=0}^a \dEE[Q](\mb s,t)
\end{align}
is a Q-Wiener
process \citep{da2014stochastic}, with spatial precision operator $Q/a$.

The case of a separable covariance model with a Mat\'ern covariance in space and an exponential covariance in time is obtained from the stationary solutions to
\begin{align}\label{eq:simplest}
\left( \frac{\partial}{\partial t} + \kappa \right) u(\s, t) = \dEE[Q](\s, t), \qquad (\mv{s},t)\in\cD\times \mathbb{R}.
\end{align}
This is a spatial generalisation of the Ornstein-Uhlenbeck processes.
We aim to produce
a space-time model with diffusive behaviour.
For this, we replace the dampening coefficient $\kappa$ in \eqref{eq:simplest}
with a power of the dampended diffusion operator $L_s$, defining a model family
of the time-stationary solutions to iterated diffusion-like processes
\begin{equation}\label{eq:spde_general_nonfractional}
\left(\gamma_t \frac{d}{dt} + L_s^{\alpha_s/2}\right)^{\alpha_t} u(\bm s, t) = \dEE[Q](\bm s, t),
\qquad (\mv{s},t)\in\cD\times \mathbb{R}.
\end{equation}
When $\cD= \mathbb{R}^d$, the space-stationary solutions are used. For compact manifolds with boundary, the operators $L_s$ and $Q$ are equipped with suitable boundary conditions on $\partial\cD$.
In total, the model has three non-negative smoothness parameters
$(\alpha_t,\alpha_s,\alpha_e)$ and three
positive scale parameters $(\gamma_t,\gamma_s,\gamma_e)$. It is not immediately obvious how the definition \eqref{eq:spde_general_nonfractional} would be interpreted for non-integer powers $\alpha_t$. However, by taking advantage of the spectral properties of the operators, we define the following model, which has an operator that more clearly allows fractional powers $\alpha_t$, as
\begin{equation}\label{eq:spde_general}
\left(-\gamma_t^2 \frac{d^2}{dt^2} + L_s^{\alpha_s}\right)^{\alpha_t/2} u(\bm s, t) = \dEE[Q](\bm s, t),
\qquad (\mv{s},t)\in\cD\times \mathbb{R}.
\end{equation}

\begin{theorem}
For $\cD= \mathbb{R}^d$, as well as for other domains where $L_s$ has well defined positive powers,
the definitions \eqref{eq:spde_general_nonfractional} and \eqref{eq:spde_general} of the Gaussian process $u(\mv{s},t)$ coincide for $\alpha_t\in\mathbb{N}$.
\end{theorem}

\begin{proof}
This can be seen by applying the techniques developed in \citet{vergara_general_2022}. Alternatively, the transfer function $G(\omega_t)$ \citep[see][Chapter 4]{LindgrenStocProc} for the temporal linear filter defined by the operator in \eqref{eq:spde_general_nonfractional} is $G(\omega_t)=(i\gamma_t\omega_t+L_s^{\alpha_s/2})^{\alpha_t}$, well-defined for positive integers $\alpha_t$, and has $|G(\omega_t)|^2=(\gamma_t^2\omega_t^2+L_s^{\alpha_s})^{\alpha_t}$. The transfer function $H(\omega_t)$ for the temporal linear filter defined by the operator in \eqref{eq:spde_general} is $H(\omega_t)=(\gamma_t^2\omega_t^2+L_s^{\alpha_s})^{\alpha_t/2}$, well-defined for positive $\alpha_t$. We see that $|G(\omega_t)|^2=|H(\omega_t)|^2$, so the spectral properties of the two process definitions coincide for positive integer $\alpha_t$ values.
\qqed
\end{proof}

It should be noted that it would be possible to give a more direct definition of the model \eqref{eq:spde_general_nonfractional} with fractional $\alpha_t$, but this
would require more sophisticated mathematical tools, which is outside the scope of this work.
The two representations make it clear that the model with $\alpha_e=0$ is a
special case of the \citet{stein2005space} model family, with with $a_1=0$ and $\alpha_1=1$ in \eqref{eq:steinspde}, and that the model of \citet{jones1997models} is obtained by setting $\alpha_e=0$ and $\alpha_t=1$ in \eqref{eq:spde_general_nonfractional}.

The use of the same spatial operator $L_s$ in the left hand side of \eqref{eq:spde_general}
as in the precision operator on the right hand side is what causes the spatial
marginalisation of the process to be Mat\'ern fields in the simplest case, as will be shown
in Section~\ref{sec:properties}.
The parameters $\alpha_t$, $\alpha_s$, and $\alpha_e$
determine the differential operator orders involved in the SPDE operator and therefore also the smoothness properties of the process, as shown in Section~\ref{sec:properties}.

The model can be further generalised by allowing the $\gamma$ parameters to
vary across space. This is most straightforward for $\gamma_s$, since that
only alters the $L_s$ operator.
For complex domains, as well as when $L_s$ is generalised to
vary across space, the resulting solutions are not space-stationary, but
still have marginal spatial properties defined by powers of $L_s$.
The practical precision construction in Section~\ref{sec:gmrfrepr} can be generalised
to separable non-stationarity, where $\gamma_t$ is allowed to depend on time
and $\gamma_s$ and $\gamma_e$ depends on space, since that retains commutativity
between the temporal and spatial operators.

\label{sec:nonstationary-Ls}

\subsection{Compact domains and manifolds}

For compact domains, the model definitions include some form of boundary conditions.
These boundary conditions induce boundary effects near the domain boundary,
and as shown in \citet{lindgren2011explicit}, if such effects are undesirable,
one can extend the domain by at least the spatial range.  By taking advantage of
a non-stationary spatial operator, the barrier method introduced by
\citet{bakka2019nonstationary} can also be used to nearly eliminate boundary effects,
as well as to obtain models that appropriately take complex geography into account.
Similarly, all the common extensions onto curved manifolds, such as the globe,
can be implemented using the same approaches as for $\RR^d$. This includes the finite
element methods used in Section~\ref{sec:gmrfrepr}, but also Fourier-like
spectral basis function expansions given by the eigenfunctions of the Laplacian,
either given in closed form, e.g.,\ spherical harmonics on the globe, or obtained
numerically from finite element eigenfunction computations. See \citet{lindgren_spde_2022}
for an overview of the literature on these alternative methods.

% !TEX root = spacetime_jrssb.tex
\section{Parameter interpretations and model properties}
\label{sec:properties}
In this section we discuss
marginal spatial and temporal properties of the diffusion-based model \eqref{eq:spde_general_nonfractional}.
In order to simplify the exposition, we focus on the ordinary Mat\'ern covariance
case when the spatial domain is $\cD= \mathbb{R}^d$.
In this case, the space-time spectral density of the stationary solutions
$u(\mb s, t)$ to \eqref{eq:spde_general} is
\begin{equation}\label{eq:spectrum}
S_{u}(\mb{\omega}_s, \omega_t) = \frac{1}{
(2\pi)^{d+1}\gamma_e^2[\gamma_t^2\omega_t^2+
(\gamma_s^2+\| \mb{\omega}_s \|^2)^{\alpha_s}]^{\alpha_t}
(\gamma_s^2 + \| \mb{\omega}_s \|^2)^{\alpha_e}},
\end{equation}
for $(\mb{\omega}_s, \omega_t) \in\mathbb{R}^d\times\mathbb{R}$.
The space-time covariance function is given by the Fourier integral
\begin{equation}\label{eq:covariance}
R_{u}(\s, t) = \int_\RR\int_{\RR^d} \exp[i (\mb{\omega}_s\cdot\s + \omega_t t)]
S_u(\mb{\omega}_s,\omega_t) \md\s \md t
\end{equation}
for spatial lags $\s$ and temporal lags $t$.

\subsection{Sample path continuity and differentiability theory}

For fields with Mat\'ern covariance functions, the degree of differentiability is encoded
in the smoothness index $\nu$. For models with space-time spectral density given by \eqref{eq:spectrum}, the marginal covariance in time is not generally of the Mat\'ern class, so we need to use more general conditions for determining the smoothness.

The differentiability of a stationary process $x(t)$, $t\in\mathbb{R}$,
is determined by the decay rate of its spectral density.
If $S(\omega) \sim \omega^{-\gamma}$ for some $\gamma>0$ for large $\omega$,
then the process is $a$ times mean square
differentiable for all $a < \frac{\gamma - 1}{2}$ \citep{stein2005space}.
For stationary Gaussian processes, stronger statements of almost sure
sample path continuity of derivatives and Hölder continuity can be made.
The technical details can be found in Section~9.3 of \citet{cramer1967book}
and \citet{scheuerer_regularity_2010}, and are summarised in
Appendix~\ref{sec:path-continuity}, including a more formal characterisation of
the smoothness index.  The results show that Gaussian processes with spectral
densities satisfying $S(\mb{\omega})\sim\|\mb{\omega}\|^{-2\nu-d}$ for some $\nu>0$
and large $\|\mb{\omega}\|$ have
smoothness index $\nu$. This means that the sample paths have almost surely
continuous derivatives of order up to and including $k=\lceil\nu\rceil-1$, and
that the derivatives of order $k$ are Hölder of index $a$ for any $0<a<\nu-k$.
Further, the sample paths are almost surely in the Sobolev spaces $W^{b,2}$
for any $b<\nu$, on finite subsets of $\RR^d$.  Although these results are derived
specifically for $\RR^d$, it is clear that sample path properties of Whittle-Mat\'ern
fields on more general but smooth domains will have similar, and usually identical,
local differentiability
properties, based on the decay rate of the the eigenspectrum of the Laplacian.
In particular, the spectral Fourier representations on the 2D sphere $\sphere^2$
lead to series that converge under the same conditions as the continuous spectra
on $\RR^2$.

The smoothness index $\nu$ can be interpreted as the smallest value
for which some form of weak continuity does \emph{not} hold.
For a process on a multidimensional domain with potentially different smoothness
in different directions, Theorem~\ref{thm:cramer} and smoothness definition in
Appendix~\ref{sec:path-continuity} will be applied to the one-dimensional marginals of
the process.

\subsection{Properties of the spatio-temporal model}

We can now show that the spatial marginals of $u(\mb s, t)$, i.e.\ for fixed $t$,
are Mat\'ern covariance fields,
given that the smoothness parameters are chosen appropriately.
To keep some notational brevity, we first define the unit variance and range Mat\'ern
covariance function $R^M_\nu(t)$,
\begin{equation}
\label{eq:matern-cov}
R^M_{\nu}(t)
=\frac{1}{\Gamma(\nu)2^{\nu-1}}t^{\nu}K_{\nu}(t),\quad t\geq 0,
\end{equation}
and the scaling constants
$$
C_{\RR^d,\alpha} = \frac{\Gamma(\alpha-d/2)}{\Gamma(\alpha) (4\pi)^{d/2}},
$$
for $d=1,2,3,\dots$ and $\alpha > d/2$. These appear as variance scaling constants
for the regular Whittle-Mat\'ern SPDE models.

\begin{proposition}\label{prop:spatial_cov}
Define the effective spatial marginal operator order
$\alpha = \alpha_e + \alpha_s(\alpha_t-\nicefrac12)$ and assume that $\alpha>\nicefrac{d}{2}$.
Then the solution $u(\mb s, t)$ to \eqref{eq:spde_general} has marginal spatial
covariance function
$$
\pCov(u(\mb{s}_1,t),u(\mb{s}_2,t)) =
\sigma^2 R^M_{\nu_s}(\gamma_s\|\mb{s}_2-\mb{s}_1\|)
$$
where $\nu_s = \alpha-d/2$ is the spatial smoothness index and
\begin{equation}\label{eq:variance}
\sigma^2 = \frac{C_{\RR,\alpha_t} C_{\RR^d,\alpha}}
{\gamma_e^2 \gamma_t\gamma_s^{2\alpha-d}}.
\end{equation}
\end{proposition}

\begin{proof}
	See appendix \ref{proof:spatial_covar}, that also includes a derivation of the
	marginal spatial cross-spectra for different time lags.
	\qqed
\end{proof}

% ??????? page 313 left side, middle
\begin{proposition}\label{prop:tempsmooth}
Assume $\alpha_t, \alpha_s, \alpha_e$ satisfy $\alpha > \nicefrac{d}{2}$.
Then the temporal smoothness index of the solutions
$u(\mb s, t)$ to \eqref{eq:spde_general}
is
	$
	\nu_t = \min\left[\alpha_t-\frac1{2}, \frac{\nu_s}{\alpha_s}  \right],
	$
and for $d=2$, the marginal temporal spectrum is
	$$
	S_t(\omega_t) \propto {}_{2}F_{1}\left(\alpha_t, \frac{\alpha_e - 1}{\alpha_s} + \alpha_t, \frac{\alpha_e - 1}{\alpha_s} + \alpha_t + 1; -\omega_t^2 \gamma_t^2 / \gamma_s^{2\alpha_s}\right),
	$$
	where ${}_{2}F_{1}$ denotes the hypergeometric function.
\end{proposition}
\begin{proof}
	See appendix \ref{proof:tempsmooth}.
	\qqed
\end{proof}

For integer values of the operator orders,
the hypergeometric function
can be expressed using elementary functions.
When $\alpha_t=\alpha_s = 2$ and $\alpha_e = 0$ for $d=2$, we obtain
\begin{equation}
S_t(\omega_t) \propto \int_{0}^{\infty} \frac{1}{(\wt{\omega}_t^2 + (1+v)^{2})^2}\md v = \frac{\arctan(\wt{\omega}_t)}{2\wt{\omega}_t^3} - \frac{1}{2\wt{\omega}_t^2(\wt{\omega}_t^2+1)},
\end{equation}
where $\wt{\omega}_t = \omega_t\gamma_t/\gamma_s^{\alpha_s}$,
showing that the marginal temporal covariance is not a Mat\' ern covariance.
%However, similarly
%to the Mat\'ern covariance,
%we can control its smoothness and range.
The exception is the separable case, where the temporal covariance function is
a Mat\'ern covariance function with smoothness index
$\alpha_t-1/2$.

\begin{corollary}\label{coroll:separable_cov}
	Assume that $\alpha_s = 0$, $\alpha_t>1/2$, and $\alpha_e>d/2$.
	Then the stationary solutions $u(\mb s, t)$ to \eqref{eq:spde_general}
	have a separable space-time covariance function where
	the spatial covariance is given by
	Proposition~\ref{prop:spatial_cov} and
	the marginal temporal covariance function is
	$$
	C(u(\mb{s},t_1),u(\mv{s},t_2)) = \sigma^2 R^M_{\nu_t}(\gamma_t^{-1} |t_2-t_1|),
	$$
	where $\nu_t = \alpha_t-1/2$ and $\sigma^2$ is given by \eqref{eq:variance}
	with $\alpha=\alpha_e$.
\end{corollary}
\begin{proof}
Follows directly from the product form of the space-time spectrum
\eqref{eq:spectrum}.
\qqed
\end{proof}

In Table~\ref{tab1},
we summarise the general smoothness results,
as well as some important special cases. The special cases denoted \emph{diffusion}
are generalised analogues of the diffusion-injection model \eqref{eq:whittlespacetime},
and the special \emph{critical diffusion} model is later used in Sections~\ref{sec:gmrfrepr}
and \ref{sec:applic}.
The general conditions on the $\alpha$ parameters that give well defined solutions are encoded in the spatial and temporal smoothness conditions $\nu_s>0$ and $\nu_t>0$, and can also be written as the conditions $\alpha=\alpha_e+\alpha_s(\alpha_t-1/2)>d/2$ and $\alpha_t>1/2$.

\begin{table}
\caption{Summary of the smoothness properties of the solutions $u(\mv{s}, t)$ for different values of the parameters $\alpha_t, \alpha_s, \alpha_e$, together with some examples. Here $\nu_t$ and $\nu_s$ respectively denote the temporal and spatial smoothnesses of the process.
		\label{tab1}}
	\resizebox{0.95\linewidth}{!}{\begin{tabular}{cccccc}
		\toprule
		$\alpha_t$ & $\alpha_s$ & $\alpha_e$ & Type & $\nu_t$  	& $\nu_s$\\
		\cmidrule(r){1-6}
		$\alpha_t$ & $\alpha_s$ & $\alpha_e$ & General & $\min\left[\alpha_t-\frac1{2}, \frac{\nu_s}{\alpha_s}  \right]$ & $\alpha_e + \alpha_s(\alpha_t-\frac12)-\frac{d}{2}$\\
		\cmidrule(r){1-6}
		$\alpha_t$ & 0 & $\alpha_e$ & 	Separable & $\alpha_t-\frac1{2}$ & $ \alpha_e -\frac{d}{2}$\\
		$\alpha_t$ & $\alpha_s$ & $\frac{d}{2}$ & Critical & $\alpha_t-\frac1{2}$ &  $\alpha_s(\alpha_t-\frac1{2})$\\
		$\alpha_t$ & $\alpha_s$ & 0 & Fully non-separable & $\alpha_t-\frac1{2} - \frac{d}{2\alpha_s}$ &  $\alpha_s(\alpha_t-\frac1{2}) -\frac{d}{2}$\\
		\cmidrule(r){1-6}
		$1$ & $2$ & $\alpha_e>\frac{d}{2}$ & Sub-critical diffusion & $1/2$ & $\alpha_e+1-\frac{d}{2}$\\
		$1$ & $2$ & $\frac{d}{2}$ & Critical diffusion & $1/2$ & $1$\\
		$1$ & $2$ & $\frac{d}{2}-1<\alpha_e< \frac{d}{2}$
		& Super-critical diffusion & $\nu_s/2$ & $\alpha_e+1-\frac{d}{2}$\\
%		\cmidrule(r){1-6}
%		$1$ & $3$ & $0$ & Fully non-separable & $1/6$ & $3/2-d/2$\\
%		$3$ & $1$ & $0$ & Fully non-separable &  $1$ & $5/2-d/2$ \\
%		\cmidrule(r){1-6}
%		$3/2$ & $0$ & $3$ & Separable & $1$ & $3-d/2$ \\
%		$3/2$ & $2$ & $1$ & Diffusion & $1$ & $3-d/2$\\
%		$2$ & $2$ & $0$ & Fully non-separable &  $1$ & $3-d/2$ \\
%		\cmidrule(r){1-6}
%		$1$ & $4$ & $0$ & Fully non-separable &  $1/4$ & $2-d/2$ \\
%		$2$ & $4/3$ & $0$ & Fully non-separable & $3/4$ & $2-d/2$ \\
%		$5/2$ & $1$ & $0$ & Fully non-separable & $1$ & $2-d/2$ \\
		\cmidrule(r){1-6}
		$1$ & $0$ & $2$ & Separable & $1/2$ & $2-\frac{d}{2}$ \\
		$3/2$ & $2$ & $0$ & Fractional diffusion &  $1-\frac{d}{4}$ & $2-\frac{d}{2}$ \\
		$2$ & $2$ & $0$ & Iterated diffusion &  $\frac32-\frac{d}{4}$ & $3-\frac{d}{2}$ \\
		\bottomrule
	\end{tabular}}

\end{table}

\subsubsection{Quantifying non-separability}

From Table~\ref{tab1} we can see that the $\alpha_e$ parameter controls the
type of non-separability. An important case is $\alpha_e = 0$,
which we refer to as fully non-separable models. The spectral density for such models
is a subfamily of the \cite{stein2005space} spectral model family.
The degree of non-separability can be quantified by the relation between
$\alpha_e$ and the effective marginal spatial operator order $\alpha$.
We introduce the non-separability parameter
%\begin{align*}
$\beta_s = 1 - \alpha_e/\alpha = 1 - \alpha_e/(\nu_s+d/2) \in [0, 1]$,
%\end{align*}
where $\beta_s=0$ gives a separable model, and
$\beta_s=1$ gives a ``maximally non-separable'' model.
Assuming given values for the temporal smoothness $\nu_t>0$, spatial smoothness $\nu_s>0$, and non-separability $\beta_s\in[0,1]$, we can find the corresponding values of $(\alpha_t,\alpha_s,\alpha_e)$. Let $\beta_*(\nu_s,d)=\frac{\nu_s}{\nu_s+d/2}$. Then
\begin{align*}
\alpha_t &= \nu_t \max\left(1, \frac{\beta_s}{\beta_*(\nu_s,d)}\right)+\frac{1}{2},\\
\alpha_s &=
\frac{\nu_s}{\nu_t} \min\left(\frac{\beta_s}{\beta_*(\nu_s,d)},1\right)
=
\frac{1}{\nu_t} \min\left[(\nu_s+d/2)\beta_s,\nu_s\right]
,\\
\alpha_e &= \frac{1-\beta_s}{\beta_*(\nu_s,d)} \nu_s = (\nu_s+d/2)(1-\beta_s).
\end{align*}
The critical branching point $\beta_s=\beta_*(\nu_s,d)$ motivates the term
\emph{critical} for such models.
Models with $\beta_s<\beta_*(\nu_s,d)$ are \emph{sub-critical} and models with
$\beta_s>\beta_*(\nu_s,d)$ are \emph{super-critical}.  The critical models have $\alpha_t=\nu_t+1/2$,
$\alpha_s=\nu_s/\nu_t$, and $\alpha_e=d/2$. The \emph{diffusion} models in
Table~\ref{tab1} with $\alpha_t=1$ and $\alpha_s=2$ are of particular interest,
as they arise from a basic heat equation. Notably, the fully non-separable
diffusion model DEMF(1,2,0) requires $d=1$ to ensure $\nu_s>0$, whereas the fully non-separable
twice iterated diffusion model DEMF(2,2,0) is valid for $d\in\{1,2,3,4,5\}$.

\subsubsection{Scale parameter interpretation}

To improve the interpretability of the scale parameters,
we define  $\sigma$, $r_s$, and $r_t$ via
\begin{align}
\sigma^2 &= \frac{C_{\RR,\alpha_t} C_{\RR^d,\alpha}}{\gamma_t \gamma_e^2 \gamma_s^{2\alpha-d}} \\
r_s &= \gamma_s^{-1} \sqrt{8\nu_s} \\
r_t &= \gamma_t \gamma_s^{-\alpha_s} \sqrt{8 (\alpha_t-1/2)},
\end{align}
where $r_s$ is the correlation range as
in \citet{lindgren2011explicit},
giving approximately correlation of 0.13 at $r_s$
distance in space (keeping time fixed).
Similarly, $r_t$ controls the temporal correlation range
for the separable model. In the non-separable cases, it is the temporal correlation range for the evolution of the spatial eigenfunction corresponding to the smallest
eigenvalue of the Laplacian, i.e.\ a constant function over space, evolving in time.
Eigenfunctions for larger spatial eigenvalues have shorter temporal correlation range,
so the combined effective range will typically be smaller than the nominal $r_t$ value
would indicate.

\subsection{Examples}
\label{sec:model-examples}

\begin{table}
\caption{Four specific DEMF models on $\RR^d$.
		\label{tab2}}
	\begin{tabular}{ccccccc}
		\toprule
		Model & $\alpha_t$ & $\alpha_s$ & $\alpha_e$ & Type & $\nu_t$  	& $\nu_s$\\
		\cmidrule(r){1-7}
		A: DEMF(1,0,2) & $1$ & $0$ & $2$ & Separable order 1 & $1/2$ & $1$ \\
		B: DEMF(1,2,1) & $1$ & $2$ & $1$ & Critical diffusion & $1/2$ & $1$\\
		C: DEMF(2,0,2) & $2$ & $0$ & $2$ & Separable order 2 &  $3/2$ & $1$ \\
		D: DEMF(2,2,0) & $2$ & $2$ & $0$ & Iterated diffusion &  $1$ & $2$\\
		\bottomrule
	\end{tabular}
\end{table}

Consider the four models on $\RR^2$ defined in Table~\ref{tab2},
and choose the $\gamma$ parameters so that $\sigma=1$ and $r_s=1$ for each model.
Further, $\gamma_t$ is chosen so that the nominal $r_t$ value is $1$, so we can
compare the non-Mat\'ern behaviour of the temporal correlation to the
spatial Mat\'ern behaviour.

	In general, the covariances are not available in closed form, but since the
	temporal covariance for each spatial frequency is of Mat\'ern type, the spatial
	cross-spectra (derived in Appendix~\ref{proof:spatial_covar}) can be inverted numerically to obtain the cross-covariance.
	Specifically, the cross-covariance can be computed numerically with a 2D fast Fourier transform (FFT)
	computation for each fixed temporal lag (see Appendix~\ref{app:spec2cov}).
	This technique is related to the
	half-spectral space-time covariance models from \citet{horrell_half-spectral_2017}. There, they
	focus on models where the temporal spectrum is known for each spatial location,
	$\mathcal{F}_t R(\s,t) = f(\omega_t)g(\s,\omega_t)$,
	but the theory also covers the case of known spatial spectrum for each time point,
	$\mathcal{F}_s R(\s,t) = f(\mb{\omega}_s)g(\mb{\omega_s},t)$, that we use here.

	In Figure~\ref{fig:covariancesB} we
	show the spatio-temporal
	covariance function for
	these four models, and the marginal spatial covariances are shown
	in Figure~\ref{fig:covariancesA}.
	There is a clear difference between the
	spatio-temporal covariances,
	even though the marginal spatial covariances are identical for the first three models.

\begin{figure}[t]
	\begin{center}
		\includegraphics[width=0.95\linewidth]{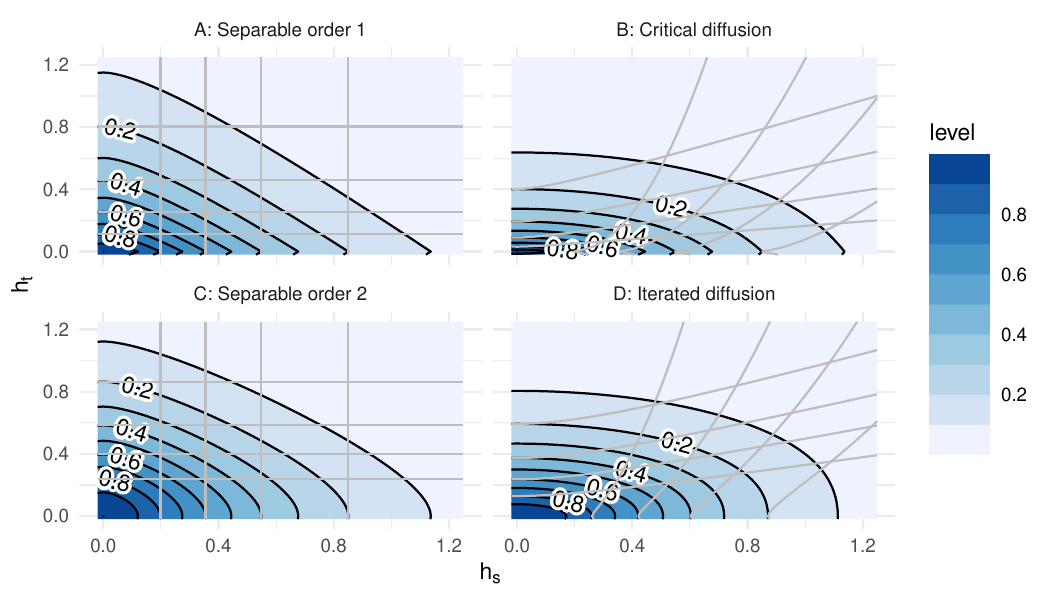}
	\end{center}
\vspace{-0.7cm}
	\caption{The space-time covariance functions for spatial dimension $d=2$,
	for the four models from Table~\ref{tab2}, Section~\ref{sec:model-examples}.
	The grey overlayed curves are level curves of the the relative decay of
	the spatial and temporal covariances in relation to the marginal covariances.
	The non-separable models have non-orthogonal decay.
	\label{fig:covariancesB}
	}
\end{figure}

\begin{figure}[t]
	\begin{center}
		\includegraphics[width=0.8\linewidth]{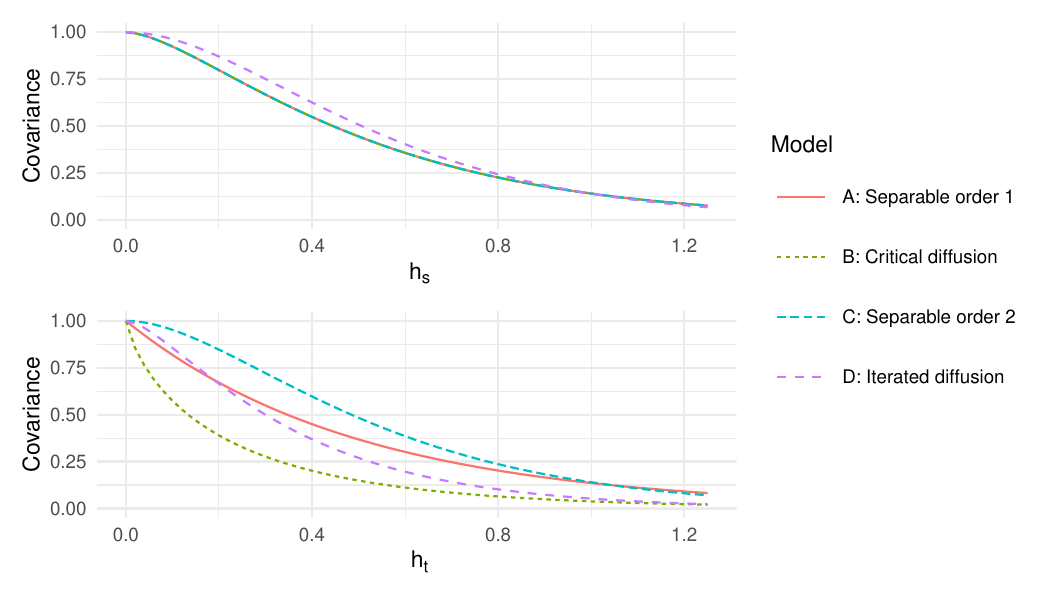}
	\end{center}
\vspace{-0.7cm}
	\caption{The marginal spatial and
		temporal covariances of for spatial dimension $d=2$,
	for the four models from Table~\ref{tab2}, Section~\ref{sec:model-examples}.
		The spatial correlation is approximately 0.13
		when the distance equals the range $r_s$. For the temporal correlations,
		that relationship to $r_t$ only holds for the contribution from the evolution
		of a spatial constant, and the effective range has a more complex structure,
		depending on the combined model parameter.
	}
	\label{fig:covariancesA}
\end{figure}

\subsection{Spheres and other manifolds}
As noted earlier, the marginal spatial covariance properties of the DEMF models
on general manifolds are rooted in the properties of the Whittle-Mat\'ern operator,
and depend on the specific geometry. However, the temporal structure is linked to each
spatial frequency in the same way for every manifold, so we can focus on the
effects on the spatial properties.
Smoothness properties intuitively follow from the local properties of the differential operator
on smooth manifolds, that locally behave like $\RR^d$, so that is not the main
obstacle to determining the process properties.
Instead, it is the effect of the manifolds intrinsic curvature that prevents general
closed form expressions for the covariance functions to be derived.
On a compact manifold $\cD$, the covariance function for models based on
$L_s^{\alpha/2}=(\gamma_s^2-\Delta)^{\alpha/2}$
(where $\alpha=\alpha_e+\alpha_s(\alpha_t-1/2)$ in the DEMF models)
take the form
$$
R(\s,\s') = \sum_{k=0}^\infty C_k \frac{1}{(\gamma_s^2+\lambda_k^2)^{\alpha}} E_k(\s)E_k(\s'), \quad \s\in \cD,
$$
where $(\lambda_k,E_k)$ are the eigenvalue/function pairs of the
$-\nabla\cdot\nabla$ (negated Laplace-Beltrami) operator on $\cD$, and $C_k$ are
scaling constants that depend on potential scaling of the eigenfunctions and
multiplicity of eigenvalues. This was used in \citet{lindgren2011explicit} to
show that the finite element constructions for Whittle-Mat\'ern fields work on general manifolds.
On the sphere, the eigenfunctions are the spherical harmonics, with eigenvalues $\lambda_k=k(k+1)$
with multiplicity $2k+1$. With the spherical harmonic definitions in
Appendix~\ref{sec:spherical-harmonics}, the resulting covariance can be simplified to
\begin{align}
\label{eq:spherical-covariance}
R_{\sphere^2,\alpha}(\s,\s';\gamma_s) &=
\sum_{k=0}^\infty \frac{2k+1}{4\pi [\gamma_s^2 + k(k+1)]^\alpha} P_{k,0}(\s\cdot\s') ,
\end{align}
where $P_{k,0}(\cdot)$ are Legendre polynomials of order $k$, and the factor $\frac{2k+1}{4\pi}$
comes from the eigenvalue multiplicity and Fourier-Bessel transform theory on the sphere
(see Appendix~\ref{sec:spherical-harmonics}).
It follows from the construction that the infinite series for the covariances of
the process derivatives that the differentiability properties on the sphere
are the same as on $\RR^2$, as the terms $\lambda_k^a\frac{2k+1}{[\gamma_s^2+\lambda_k]^\alpha}$
decay at the same rate as required for the smoothness criteria on $\RR^2$ from
Appendix~\ref{sec:path-continuity}.

Due to the wraparound effects on the sphere, the spatial variance contribution to the
overall field variance is not the same as on $\RR^2$, and the factor $C_{\RR^d,\alpha}/\gamma_s^{2\alpha-d}$ in \eqref{eq:variance} needs to be replaced by a function of $\gamma_s$ defined by
\begin{align}
\label{eq:spherical-variance}
C_{\sphere^2,\alpha}(\gamma_s) &= \sum_{k=0}^\infty \frac{2k+1}{4\pi [\gamma_s^2 + k(k+1)]^\alpha} ,
\end{align}
obtained from the spectral representation of a spherical Whittle-Mat\'ern field.
The overall variance can then be written as
$\Var[u(\s,t)] =
\frac{C_{\RR,\alpha_t}}{\gamma_e^2\gamma_t}
 C_{\sphere^2,\alpha}(\gamma_s),$
and the asymptotic behaviour of $C_{\sphere^2,\alpha}(\gamma_s)$ as $\gamma_s$ approaches $0$ or $\infty$ is given by
\begin{align*}
C_{\sphere^2,\alpha}(\gamma_s) =
\sum_{k=0}^\infty \frac{2k+1}{4\pi [\gamma_s^2 + k(k+1)]^\alpha}
&\sim
\begin{cases}
\frac{1}{4\pi \gamma_s^{2\alpha}}, & \gamma_s\rightarrow 0, \\
\frac{1}{4\pi (\alpha-1) \gamma_s^{2\alpha-2}}, & \gamma_s\rightarrow \infty.
\end{cases}
\end{align*}
This shows that for large $\gamma_s$, i.e.\ short spatial ranges,
the variance of the field $u(\s,t)$ on the sphere is the same as on $\RR^2$, but
for small $\gamma_s$, i.e.\ long spatial ranges, the spherical geometry leads
to larger variance than on $\RR^2$.
For intermediate $\gamma_s$ values, the upper tail of the infinite series can be
bounded by tractable integrals, which also allows bounding the relative error
in numerical covariance and variance evaluation,
by replacing the upper series tail from $k=K$ by the integral
$
\int_{K+1/2}^\infty \frac{2k+1}{4\pi [\gamma_s^2 + k(k+1)]^\alpha} \md k .
$
More details are given in Appendix~\ref{sec:spherical-variance-approximation}.

\begin{comment}
\section{Penalized complexity priors}
\label{sec:pcprior}
In this section we use the interpretable parameters
to define penalized complexity priors \citep{simpson2017penalising}.
We first fix $r_t$, and
construct a prior for the other hyper-parameters
based on the marginal spatial structure.
Here, we use the penalized complexity priors
developed by
 \citet{fuglstad2015interpretable},
\begin{align}
\sigma &\sim \pExp(\lambda_e) \\
r_s^{-d/2} &\sim \pExp(\lambda_s)
\end{align}
where the $\lambda_s$ and $\lambda_e$ are found by eliciting prior information about a quantile, see \citet{fuglstad2015interpretable} and the Appendix of \citet{bakka2019nonstationary}.

The prior for $r_t$ we define conditionally on $\sigma$ and $r_s$.
In the separable case ($\alpha_s=0$)
we suggest the prior
\begin{align}
1/\sqrt{r_t} \sim \pExp(\lambda_t),
\end{align}
because that is the PC prior for the
marginal temporal model  \citep{fuglstad2015interpretable} in the separable case, and for the evolution of a spatial constant in the non-separable cases.

Even though the effective temporal correlation range has a more complex relationship
than the basic interpretation of $r_t$ (see Figure~\ref{fig:covariancesA}),
a more direct link to the effective correlation
range would require a more involved equation for $r_t$, which is not yet available.

\end{comment}

% !TEX root = spacetime_jrssb.tex
\section{Hilbert space representation}
\label{sec:gmrfrepr}

The discussion up to this point has focused on the general continuous domain properties of the
proposed model class.  We will now discuss aspects of numerical implementations,
suitable for inclusion in generalised additive latent Gaussian models, as available in the
\texttt{INLA} and \texttt{inlabru} packages for \texttt{R}.  The general construction
is appliccable to a wide range of basis function representations.  In practice,
we will use the finite element approach from \citet{lindgren2011explicit} due it's
computational convenience, in particular in the unstructured spatial observation location
and manifold domain contexts.

\subsection{Hilbert space approximation}

We consider general Kronecker product basis expansions
\begin{equation}
u(\s,t) = \sum_{i=1}^{n_s}\sum_{j=1}^{n_t}\psi_i(\s)\phi_j(t) u_{ij},
\label{eq:kronecker-expansion}
\end{equation}
where $\{\psi_i(\s);i=1,\dots,n_s\}$ and $\{\phi_j(t);j=1,\dots,n_t\}$
are finite basis sets for Hilbert spaces on a spatial domain $\cD$ and a time interval $[T_0,T_1]\subset\RR$, respectively.  We will show that
projection onto the resulting Kronecker function space only involve
integrals of the form $\<\phi_j,\phi_{j'}\>$,
$\<(-\Delta)^{k/4}\phi_j,(-\Delta)^{k/4}\phi_{j'}\>$,
and $\<L_s^{k/2}\psi_i,L_s^{k/2}\psi_{i'}\>$.
This is possible due to the lack of interaction in the individual model operators;
the operator as a whole is non-separable, but each operator term is space-time separable.
This also extends to the case of a non-stationary $L_s$ operator, as mentioned
in Section~\ref{sec:nonstationary-Ls}.

Different choices of spatial and temporal basis functions have benefits and drawbacks
depending on the specific modelling and data context.
%, and we will now discuss some aspects of this.
A natural choice for the spatial domain is local piecewise linear basis functions. Such functions were used in \citet{lindgren2011explicit} to construct model representations with sparse precision matrix structure
for the basis expansion coefficients, via Gaussian Markov random fields (GMRF).
This allows a large number of basis functions to be used,
and pointwise georeferenced observations will not alter the sparseness of the posterior
precision matrix, making this a versatile approach, that can also be used in combination with sparse
matrix solvers developed for ordinary deterministic PDE computations.
For very smooth processes, the piecewise linear basis functions can in principle
be replaced by higher order local polynomials \citep{liu_efficient_2016}, but
this can be difficult to implement. For non-stationary $L_s=\gamma_s(\s)-\Delta$, the spatially varying $\gamma_s(\s)$ values
only have a local influence on the finite element construction, so the additional computational
complexity lies mainly on the increased number of parameters needed to represent
the spatial variation of $\gamma_s(\cdot)$.

An alternative to piecewise linear basis functions are Harmonic basis functions based on the eigenfunctions of the Laplacian. These can be very efficient on domains that admit fast Fourier inversion algorithms,
such as $\RR^d$ and partially on $\sphere^2$.  However, the diagonal precision
matrix structure implied by the basic models is broken by scattered georeferenced
observations, as the resulting posterior precision matrix becomes dense, so
the utility is greatest for very smooth processes that can cut off the
harmonics at a long spatial range.
So-called \emph{conditioning by kriging} can also be applied in such cases, but this
is computationally expensive for large numbers of observations unless the number of basis functions
is kept small.
A further complication on general domains and manifolds is the lack of closed form
expressions for the harmonics. Computing them with e.g.\ finite element methods
is as expensive as applying the piecewise linear basis GMRF representations directly.
They are also impractical for non-stationary operators, since the precision matrices
will typically become dense instead of diagonal.

A third alternative is Karhunen-Lo\`eve expansions, which yield better approximations for fewer basis functions than harmonic basis. They can handle non-stationary operators, but needs recomputing the basis for each set of parameter values, making inference expensive.  For irregular data, it has the same problem of turing a sparse
prior precision matrix into a dense posterior precision matrix.
However, for given parameters, it can in principle be applied to the posterior distribution instead.
Unfortunately, the numerical computations for each eigenfunction is at least as expensive as computing the posterior expectation using the \emph{same} numerical method (e.g., finite elements) as in the GMRF computations, making the full computation much more expensive, and best suited to special cases such as computing a compact representation of a given, fixed, distribution.

Despite their practical numerical cost and other related problems, the harmonic basis and K-L expansions are excellent tools for theoretical analysis, and their discrete domain formulations are essential in the theoretical proofs of the general discretisation construction below.
See \citet{lindgren_spde_2022} for further discussion on the relative merits of
different basis choices.

The above considerations largely apply to the temporal basis function choice as well,
with a few useful differences. First,  in addition to piecewise linear basis functions, B-spline basis functions
of higher order can readily be applied, and in particular second order B-splines
(piecewise quadratic basis functions) provide immediate benefits with only minimal extra effort.
Where piecewise linear basis functions require some form of mass lumping for operator order 2,
second order B-splines can be applied with least squares finite element projection,
and the resulting discretised Laplacian operator matrix has the same non-sparsity as for piecewise linear basis functions. In addition, when applied to order $1$ operators, temporal interpolation
in the finite dimensional representation exhibits less quasi-deterministic fluctuations
than for piecewise linear basis functions.
Second, Harmonic basis functions are useful for smooth cyclic processes, e.g.\ seasonal effects,
but otherwise suffer from the same issues as in space.

\subsection{Precision matrix construction}

In this section we
represent the stochastic processes DEMF($\alpha_t$,$\alpha_s$,$\alpha_e$)
using general Kronecker basis Hilbert space representations.
Define $u(s,t)$ on $\cD \times \RR$,
for some polygonal domain
${\Omega \subset \RR^d}$, as the solution to \eqref{eq:spde_general}
with some boundary conditions on $\partial\cD$.
The particular choice of boundary conditions does not matter much in what follows
as long as they lead to a well defined precision operator for the solutions of
the equation posed on the bounded domain. However, in most practical situations one would use
homogeneous Neumann boundary conditions on the spatial domain.
For implementations, we restrict the temporal domain to an interval, and we then also need to impose boundary temporal boundary conditions. However, temporal boundary effects can be handled by direct calculations for the resulting AR(2)
dependence structure for the temporal coefficients in the approximation; see Appendix~\ref{sec:temporal-disc}.

The projection of the solutions onto the finite Hilbert space result
in a discretised model where the coefficients $u_{ij}$ in \eqref{eq:kronecker-expansion}
have a precision matrix that is expressed as a sum of kronecker products.
As in \citet{lindgren2011explicit}, the approximation properties of the discretisation
is directly linked to the expressiveness of the finite dimensional Hilbert space
spanned by the kronecker basis $\{\psi_i(\s)\phi_j(t),\,i=1,\dots,n_s,\,j=1,\dots,n_t\}$.

We provide the following theorem that links the continuous domain DEMF models to
finite dimensional Hilbert space representations. The theorem focuses on
the link between the continuous domain precision operator and the precision matrix,
which necessarily assuming unique solutions with a unique covariance function.
This in principle makes it applicable to more esoteric models involving various forms of intrinsic stationarity, i.e.\ non-stationary models with stationary properties with respect to some contrast filters. However, the details of such models is beyond the scope of the presentation.

\begin{theorem}\label{thm:spacetime-construction}
Let $\alpha_t\in\mathbb{N}$ and consider the equation
\begin{align}
        \label{eq:theorem-general-spde}
        \left(-\gamma_t^2\frac{\partial^2}{\partial t^2} + L_s^{\alpha_s}
        \right)^{\alpha_t/2} u(\s,t) &= \dEE[\gamma_e^2L_s^{\alpha_e}](\s,t) \quad \text{on } \cD \times [T_0,T_1],
\end{align}
where $[T_0,T_1]\subset\RR$ is a bounded interval, $L_s$ is some spatial differential operator, and some boundary conditions on $\partial \cD$ and at $T_0$ and $T_1$ are assumed such that the
precision operator for the solutions of \eqref{eq:theorem-general-spde} is well defined.
Let $\{\psi_i(\s),i=1,\dots,n_s\}$ and $\{\phi_j(t),j=1,\dots,n_t\}$ be bases for finite dimensional Hilbert spaces on $\cD$ and  $[T_0,T_1]$, respectively, chosen such that the product basis set $\{\psi_i(\s)\phi_j(t),\,i=1,\dots,n_s,\,j=1,\dots,n_t\}$ form a basis for a finite dimensional Hilbert space $V_h \subset V$, and let
$u(\s,t)=\sum_{i,j} \psi_i(\s)\phi_j(t) u_{i,j} \in V_h$ be a finite dimensional representation of a solution to \eqref{eq:theorem-general-spde}.
Assume the following two conditions:
\begin{enumerate}
\item[(i)]
Let $v(t)=\sum_{j=1}^{n_t} \phi_j(t) v_j$ be a finite dimensional approximation of a solution to
\begin{align*}
b^{1/2}\left(-\frac{\partial^2}{\partial t^2} + \kappa^2\right)^{\alpha_t/2} v(t) = \W(t), \quad \text{on } [T_0,T_1],
\end{align*}
for some $b>0$, $\kappa>0$, and $\alpha_t=1,2,\dots$, and the boundary conditions at $T_0$ and $T_1$.
Assume that the precision matrix for the weights vector $\mv{v}=(v_1,\dots,v_{n_t})$
takes the form
\begin{align*}
& b\sum_{k=0}^{2\alpha_t} \kappa^{2\alpha_t-k} \mb{J}_{\alpha_t,k/2}
\end{align*}
for some symmetric matrices $\mv{J}_{\alpha_t,0}$, $\mv{J}_{\alpha_t,1/2}$, to
$\mv{J}_{\alpha_t,\alpha_t}$.
\item[(ii)]
Let $w(\s)=\sum_{i=1}^{n_s} \psi_i(\s) w_i$ be a finite dimensional approximation of a solution to
\begin{align*}
L_s^{a/2} w(\s) &= \W(\s) \quad \text{on } \cD,
\end{align*}
where $L_s$ is equipped with the boundary conditions on $\partial\cD$, for some $a\geq 0$.
Assume that the precision matrix for $\mv{w}=(w_1,\dots,w_{n_s})$ is
$\mb{K}_a = \mv{C}^{1/2}\left(\mv{C}^{-1/2}\mv{K}_1\mv{C}^{-1/2}\right)^a \mv{C}^{1/2}$
for some symmetric positive definite matrix $\mv{K}_1$.
\end{enumerate}

Assume additionally that the temporal precision construction in condition (i) is valid for all $\kappa\geq\lambda_0^{\alpha_s/2}/\gamma_t$, where $\lambda_0$ is the smallest eigenvalue in the generalised eigenvalue problem
$\mv{K}_1\mv{e}=\mv{C}\mv{e}\lambda$.
Then, the precision matrix for the collected coefficient vector $\mv{u}=(u_{1,1},u_{2,1},\dots)$ is given by
\begin{align*}
\mb{Q}_{\mb{u}} &=
\gamma_e^2
\sum_{k=0}^{2\alpha_t}
\gamma_t^k
\mb{J}_{\alpha_t,k/2} \otimes \mb{K}_{\alpha_s(\alpha_t-k/2)+\alpha_e} .
\end{align*}
\end{theorem}
\begin{proof}
The result follows from discretising the spatial dimension, diagonalising the
resulting operator matrices, and applying the temporal precision structure condition
to the resulting independent temporal equations.
A detailed proof is given in Appendix~\ref{app:spacetime-construction}.
\qqed
\end{proof}

The existence of finite dimensional representations fulfilling conditions (i) and (ii) for certain choices of basis functions follows directly from the general constructions in \citet{lindgren2011explicit}.

For the regular Whittle-Mat\'ern operator $L_s=\gamma_s^2-\Delta$ on $\cD$ we have $\mv{K}_1=\gamma_s^2\mv{C}+\mv{G}$. For triangulated domains with local piecewise linear basis functions with $\sum_{i=1}^{n_s}\psi_i(\s)\equiv 1$ on $\cD$, we can take $\mv{C}$ to be the diagonal mass lumped mass matrix with $C_{i,i}=\<\psi_i,1\>$ and symmetric sparse structure matrix $\mv{G}$ with $G_{i,j}=\<\nabla\psi_i,\nabla\psi_j\>$. For domains where the orthogonal harmonic eigenfunctions of $\Delta$ are available, such as rectangular subdomains of $\RR^d$ and spherical harmonics on $\sphere^2$, the full mass and structure matrices $\mv{C}$ and $\mv{G}$ are both diagonal, with $C_{i,i}=\<\psi_i,\psi_i\>$.

In the temporal case, the same technique applies, but higher order B-spline basis functions are more easily applied, allowing, e.g., 2nd order B-splines to be used without mass lumping.
For temporal Neumann boundary conditions, $\mb{J}_{\alpha_t,k/2}=\mv{0}$ for odd $k=1,3,\dots,2\alpha_t-1$
and $[\mv{J}_{\alpha_t,k/2}]_{i,j}=\<(-\Delta)^{k/4}\phi_i,(-\Delta)^{k/4}\phi_j\>$ (or non-conformal approximations for non-smooth basis functions) for
even $k=0,2,\dots,2\alpha_t$.
Lemma~\ref{lemma:ar2boundary} in Appendix~\ref{sec:temporal-disc}
can be used
for 1st and 2nd order B-spline basis functions for $\alpha_t=1$ and $2$
to provide approximate stationary boundary conditions by modifying
the $\mb{J}_{\alpha_t,k/2}$ matrices for $k=0,1,\dots,2\alpha_t$.
When such temporal boundary corrections are used, fractional orders appear in $\mv{K}_{\alpha_s(\alpha_t-k/2)+\alpha_e}$ for odd $k$ unless $\alpha_s$ is an even integer. For the spatial piecewise linear finite element constructions, this would break sparsity, but for orthogonal harmonic function representations, $\mv{K}_a$ is diagonal for all $a\geq 0$, allowing the fractional powers to be used without loss of the diagonal property.

In the proof of Theorem~\ref{thm:spacetime-construction}, we see that it is sufficient that
the initial temporal precision structure is valid for
$\kappa \geq \kappa_0=\gamma_s^{\alpha_s/2}/\gamma_t$. By taking a Taylor expansion
for the boundary precision elements with respect to $\kappa$ and $\kappa_0$,
the approximation would be improved, compared with taking the Taylor expansion at $\kappa=0$,
as the expansion would be closer to the exact expression for a wider range of relevant temporal frequencies.
This improvement would however come at the expense of making the matrix constructions dependent on the $\gamma_s$ and $\gamma_t$ parameters directly.

\section{Applications}
\label{sec:applic}

\subsection{Separable vs non-separable forecasting}

The difference between using separable and the non-separable models is
most clearly seen when doing forecasting. To illustrate this, we
simulated spatial data for time $t=0$, and compute the posterior
conditional expectation for
$t=0$, $1$, and $2$ and $t=2$. For the simulation, we used a Mat\'ern
model with spatial smoothness $\nu_s=1$, matching models $\Ma$, $\Mb$, and $\Mc$,
and add one percent (standard deviation) nugget effect.
The parameters were set to $r_s=4.0$, $r_t=2.5$ for the separable models
and $r_t=4.5$ for the non-separable models, and $\sigma=1$.
The scaling difference for $r_t$ in the non-separable models
compensates for the difference in parameter interpretation
illustrated in Section~\ref{sec:model-examples}.
In the estimation, the nugget precision and the temporal range parameters $r_t$ were
kept fixed, so that only the marginal standard deviation $\sigma$ and the
spatial range parameter $r_t$ were estimated for each model.

\begin{figure}[t]
    \begin{center}
        \includegraphics[width=0.75\linewidth]{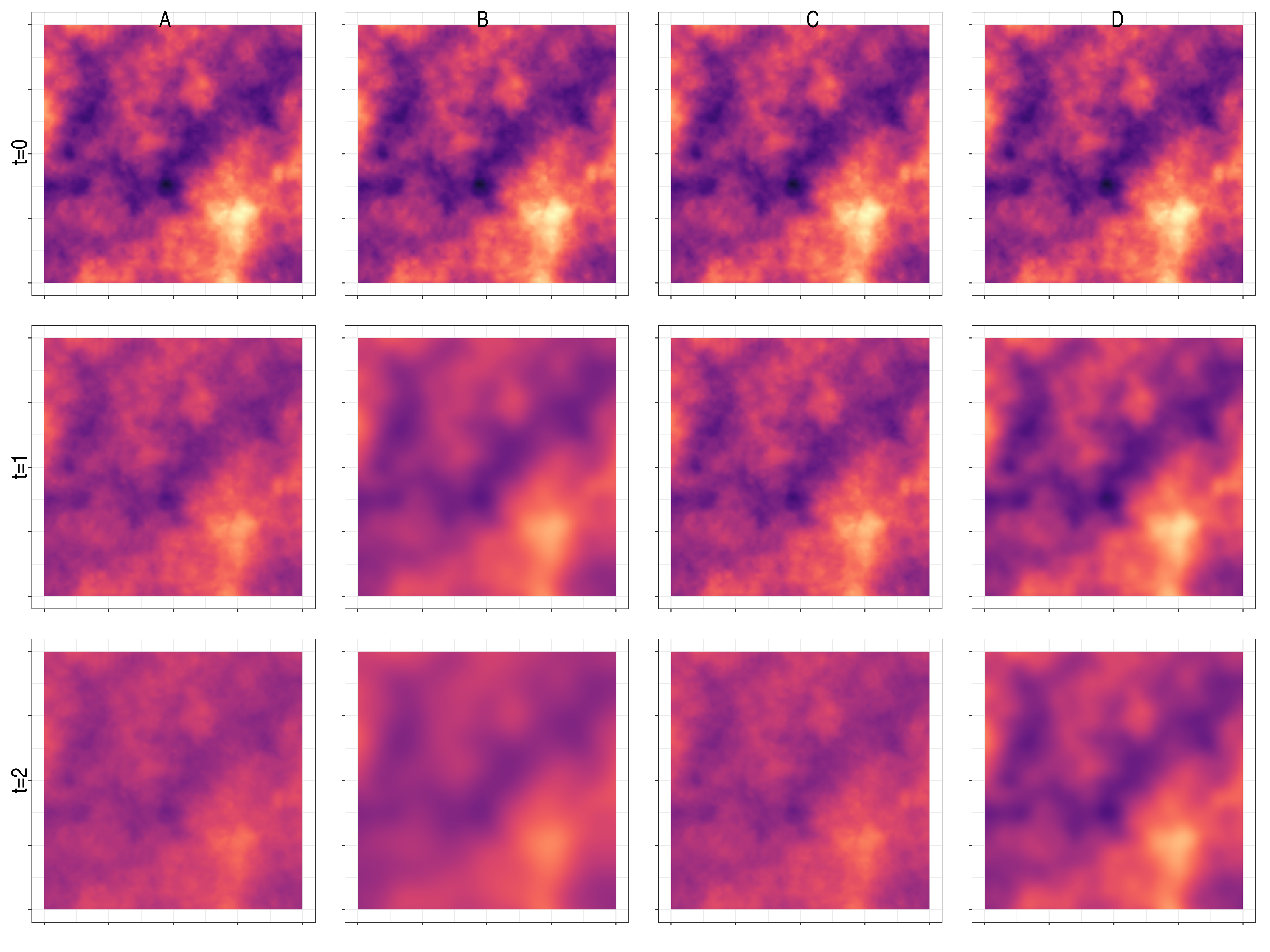}
    \end{center}
\vspace{-0.5cm}
        \caption{Predictions from each model
                when conditioned on a spatially dense
                dataset at $t=0$, and no observations for $t=1$ or $t=2$.}
        \label{fig:forecasts}
\end{figure}

Figure~\ref{fig:forecasts} displays the predictions from the four
models in Table~\ref{tab2}. For $t=0$, the results are similar for the four models, due
to the highly informative data. For the predictions for $t=1$ and $t=2$, we see
how the separable model $\Ma$ and $\Mc$ only reduce the fields point-wise towards zero,
and that non-separable models exhibit spatial diffusion, as expected. This
behaviour was part of the theoretical motivation of \cite{whittle1954,
    whittle1963stochastic}, and also a major motivation for developing the
DEMF family.
It's also noteworthy that since the forecasts are conditional expectations based
on a finite set of observations, they are smoother than the process realisations.
For the separable models, this effect isn't visible, since there this effect
only appears on smaller spatial scales than shown, but it is clearly visible
for the non-separable model. In all four cases, the posterior process realisations
however have their ordinary, lower, smoothness. This is important to take into
account when considering probabilistic forecasts, in particular for prediction
of non-linear functionals of the process.

\subsection{Global temperature dataset}
\label{sec:globaltemp}

This section presents some results analysing daily temperature data,
where all the code for the data cleaning, model fitting and plots are
included in the supplementary material.

\subsubsection{Data and model structure}

\begin{figure}[t]
        \begin{center}
            \includegraphics[width=0.49\linewidth]{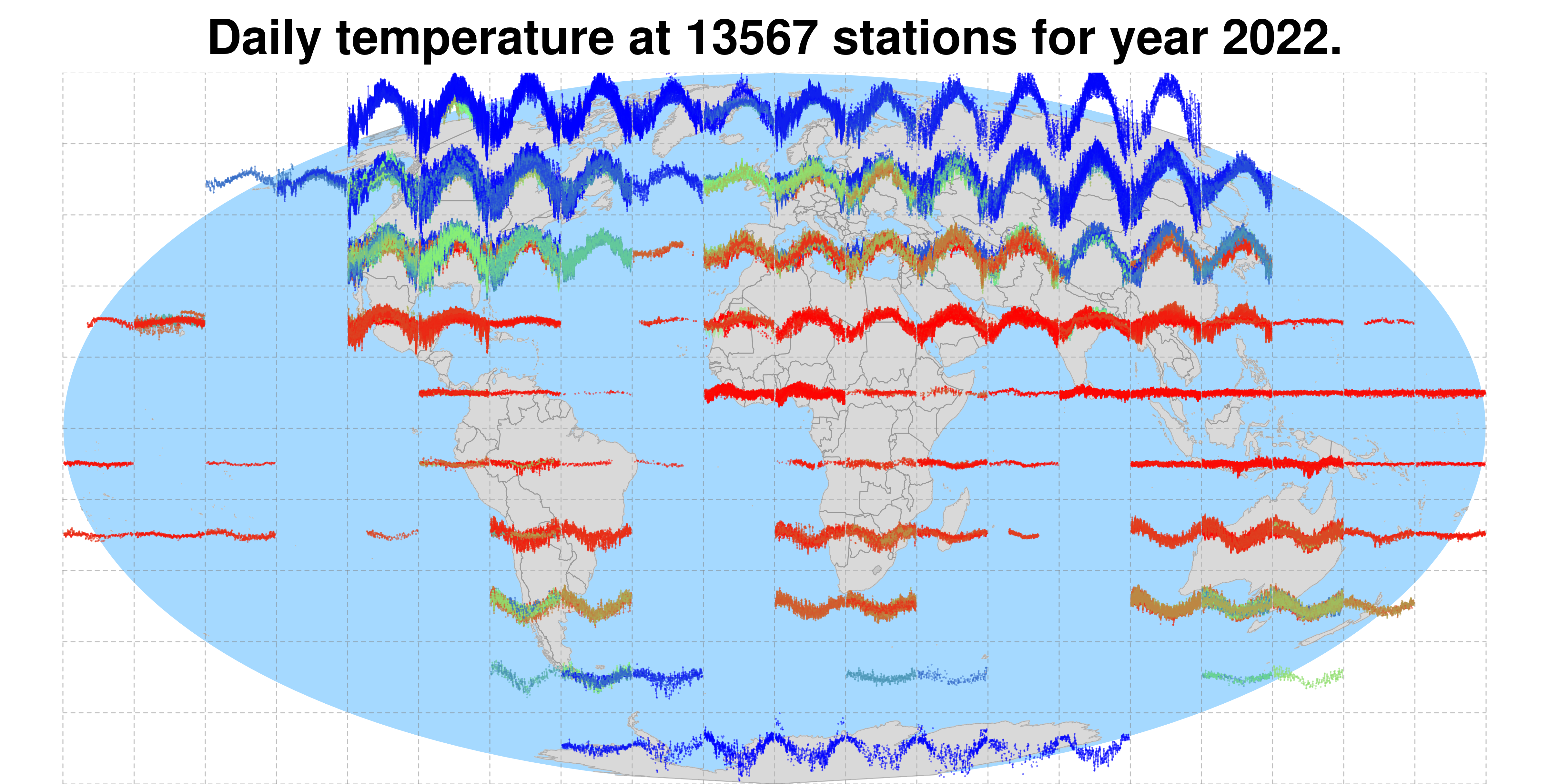}%%\\
%%            \vspace*{3mm}
            \includegraphics[width=0.49\linewidth]{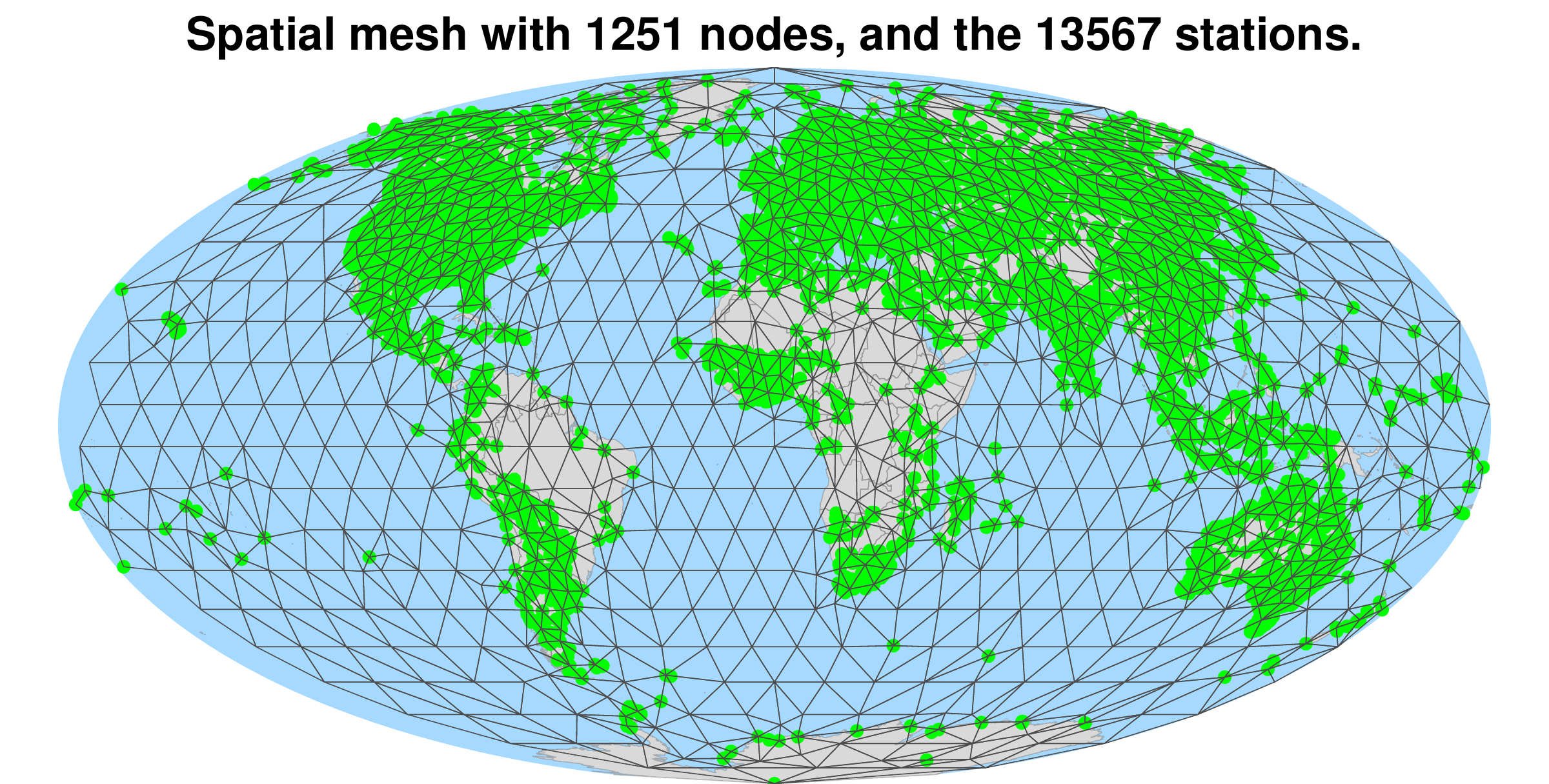}
        \end{center}
        \caption{Daily average temperature time series shown grouped
            near the corresponding locations (left), with colours based on the year average, for each station, from blue (cold), to red (warm). The locations (green) and the mesh used for the spatio-temporal model components $v$ and $u$ (right).
            The figures use a Mollweide projection,
            but the computational meshes are defined by spherical triangles directly on the globe surface.}
        \label{fig:gtdata}
\end{figure}

We used daily data for year 2022, using minimum (TMIN) and
and maximum (TMAX) daily temperatures, as described in \cite{ghcn2012data}.
We cleaned the data for inconsistencies before the analysis.
In particular, values beyond $7$ standard deviations
from the mean were treated as missing.
We computed the mean of these two variables
for each day at each one of $13\,567$ stations world-wide,
a total of $4\,951\,955$ data entries.
Figure~\ref{fig:gtdata} (top) shows this data as time series
grouped by location.

The model includes an overall level, $\mu$, the elevation in kilometres,
a smoothed deviation from the overall mean jointly over latitude and time, $b(\s,t)$,
a spatio-temporal random field $v(\s,t)$ varying slowly in time,
and a spatio-temporal random field, $u(\s,t)$, capturing the daily variability.
The $b(\s,t)$ function is allowed to vary by latitude and time,
but is fixed to zero at the equator.
The linear predictor expression is
\begin{align}
  \eta(\s,t) &= \mu + \alpha E(s) + b(\s,t) + v(\s,t) + u(\s,t).
  \label{eq:temp-eta}
\end{align}
Each observation $y_{i}$ is modelled with additive Gaussian noise with
a common variance parameter, $\sigma^2_e$, so that $y_{i}=\eta(\s_{i}, t_{i}) + e_{i}$,
where $(\s_{i},t_{i})$ is observation $i$, $i=1,\dots n$,
and $e_{i}\sim\pN(0,\sigma^2_e)$.

\subsubsection{Model discretisation and estimation}
For the $b(\s,t)$ and $v(\s,t)$ functions in the predictor expression \eqref{eq:temp-eta},
we defined temporal basis functions
$1$, $\cos[(t-1) \cdot 2\pi/ 365]$, and $\sin[(t-1) \cdot 2\pi/ 365]$.
For $b(\s,t)$, these were multiplied with two quadratic basis function
in $\sin(\text{latitude}\cdot\pi/180)$, which guarantees smooth behaviour with respect
to the location, $\s$, at the two poles, giving a total of six basis functions.
For $v(\s,t)$, each of the three temporal basis functions were instead
multiplied by stationary spatial Whittle-Mat\'ern fields over the sphere
forming a model term that captures the seasonal local deviation
from the basis seasonal pattern described by $b(\s,t)$.

The reported results were estimated using a spatial mesh with $1\,251$ nodes
(median node distance $\sim 587$km), shown in Figure~\ref{fig:gtdata},
both for the spatial coefficients in $v$ and for $u$.
For $u$, we discretised the time domain with first order basis functions
with one knot per day.
This setting gives a spatio-temporal model for $u$ of size of $456\,615$.
In vector form we have
\begin{align*}
  \mb{y} = \bm{1}\mu + \bm{E}\alpha + \mb{B}\mb{b} + \bm{A}_v\mb{v} + \bm{A}_u\mb{u} + \mb{e},
\end{align*}
where $\mb{B}$ is a six-column matrix of the evaluated basis functions for $b(\s,t)$
at the observation locations and times, and $\bm{A}_v$ and $\bm{A}_u$
contains the evaluated basis functions, respectively, for $v(\s,t)$ and $u(\s,t)$.
The vectors $\mb{b}$, $\mb{v}$, and $\bm{u}$
contain the corresponding basis weights.

We used independent priors for all the model parameters.
We used a flat prior for $\mu$ and a Gaussian with
mean zero and variance 100 for $\alpha$ and each element in $\mb{b}$.
The three spatial fields in $\mb{v}$ are assumed as independent realizations
each one modelled using Eq.~\ref{eq:simpleSpde} with a common spatial range $r_v$,
and common marginal variance $\sigma_v^2$.
The $u(\s,t)$ term is a spatio-temporal field using
one of the four models in Table~\ref{tab2}.
In total, we have six variance/range parameters to estimate.
%%%%
We used penalized complexity priors for all these parameters~\citep{art631,art590},
applied to the marginal properties of the models.
%%Section~\ref{sec:pcprior}.
%%
To define the PC-prior for $\sigma_e$ we used
$\pP(\sigma_e \geq 5) = 0.01$ and the same for $\sigma_v$ and $\sigma$.
We used $\pP(r_v \leq 600\text{km})=0.01$ for
$r_v$ and $r_s$, and for $r_t$ we used
$\pP(r_t \leq 1\text{ days})=0.01$ in models $A$ and $C$
and $\pP(r_t \leq 2\text{ days})=0.01$ in models $B$ and $D$.

\subsubsection{Model fitting results}

Attributing the relative contributions to each model component is non-trivial
due to the posterior correlation between the components. However, a basic
linear model variance decomposition,
SQT = $\sum_i(y_i -\overline{y})^2$ and SQR = $\sum_i(y_i - \pE(\eta_i\mid \mv{y}))^2$,
%$R^2=1-\pE(\sigma_e^2\mid \mv{y})/\pVar(y_{i})$, where
%$\pVar(y_{i})$ is the empirical marginal variance for the temperature variable,
can be otained to define $R^2 = 1 - \text{SQR/SQT}$.
We have that the predictor model $\eta(\s,t)$ captures 97.18\% of the variability with model $\Mb$.
Table~\ref{tab1gt} reports DIC, WAIC, and goodness-of-fit statistics for within-sample and leave-one-out
assessment \citep[leave-one-out log predictive density score, LCPO, see][]{Heldetal2010cpo},
for each of the five fitted models. For within-sample assessment, $R^2$,
mean squared error (MSE), and mean absolute error (MAE)
assess the posterior mean and median only, whereas the
log predictive density score (LPO), CRPS, and SCRPS assess the full predictive distribution
\citep{gneiting2005,bolinWallin2022scrps}.
The model $\Mo$ includes the fixed effects and $v(\s,t)$,
whereas models $\Ma$, $\Mb$, $\Mc$ and $\Md$ all include
$\mb{u}$, using the four models in Table~\ref{tab2}.
When considering $R^2$, LPO, MSE, MAE, CRPS and SCRPS
model $\Mb$ performed a slightly better.
When considering DIC, WAIC and LCPO
model $\Mc$ was slightly better.

\begin{table}
    \caption{
        Summary statistics for each estimated model.
        The LPO is the average negated log-predictive density and
        the LCPO is its leave-one-out predictions.
        The MAE, MSE, CRPS and SCRPS scores are
        in-sample statistics based on a Gaussian approximation
        of the posterior predictive distribution
        for each data point.}
        \label{tab1gt}
    \centering
\begin{tabular}{rrrrrrrrrr}
  \hline
Model & $R^2$ & DIC & WAIC & LPO & LCPO & MSE & MAE & CRPS & SCRPS \\
  \hline
$\Mo$ & 0.8663 & 5.8094 & 5.8091 & 2.9042 & 2.9046 & 19.4963 & 3.3221 & 2.4277 & 1.7903 \\
  $\Ma$ & 0.9718 & 4.3214 & 4.3206 & 2.1330 & 2.1573 & 4.1138 & 1.4660 & 1.0944 & 1.3946 \\
  $\Mb$ & \textbf{0.9718} & 4.3216 & 4.3215 & \textbf{2.1329} & 2.1575 & \textbf{4.1134} & \textbf{1.4656} & \textbf{1.0941} & \textbf{1.3945} \\
  $\Mc$ & 0.9718 & \textbf{4.3209} & \textbf{4.3192} & 2.1334 & \textbf{2.1571} & 4.1187 & 1.4675 & 1.0951 & 1.3949 \\
  $\Md$ & 0.9718 & 4.3217 & 4.3214 & 2.1331 & 2.1576 & 4.1151 & 1.4659 & 1.0944 & 1.3946 \\
   \hline
\end{tabular}
\end{table}

\begin{figure}[t]
        \begin{center}
            \includegraphics[width=0.99\linewidth]{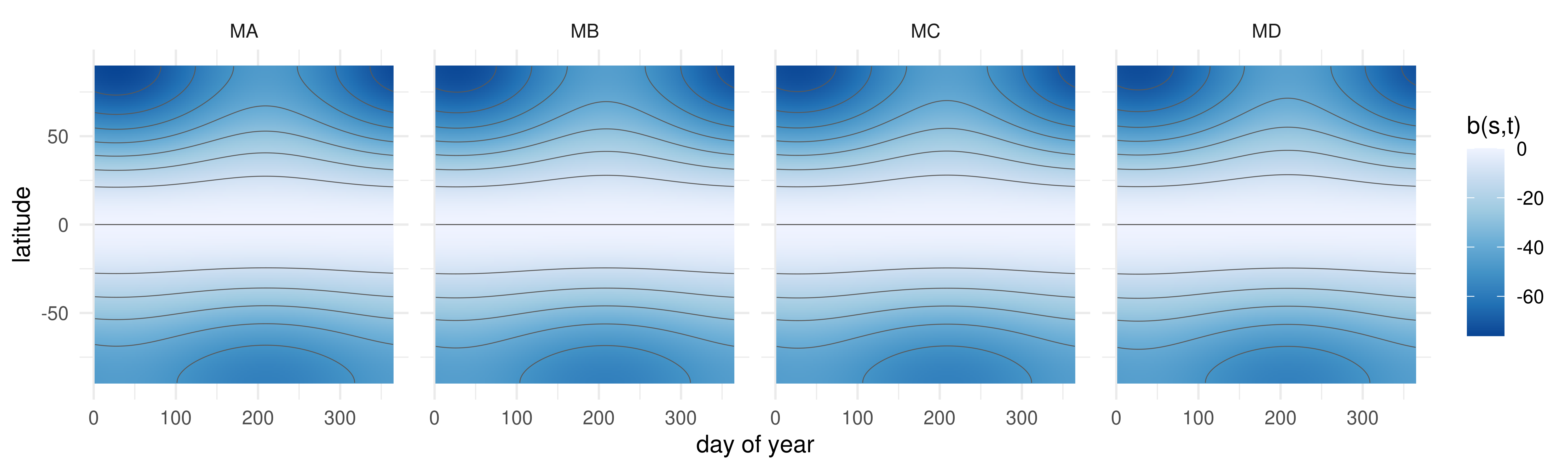}
        \end{center}
\vspace{-0.5cm}
        \caption{
            The posterior mean of the smoothed seasonal latitude effect $b(\s,t)$ for each model.
        }
        \label{figseasonallatitude}
    \end{figure}

\begin{table}
    \caption{
    The posterior mean and standard deviation (in brackets) for each of the model
    hyper-parameters.
        \label{tabhysummar}}
    \centering
\begin{tabular}{lrrrrrr}
  \hline
 & $\sigma_e$ & $r_v$ & $\sigma_v$ & $r_s$ & $r_t$ & $\sigma$ \\
  \hline
$\Ma$ & 2.07 (0.001) & 2363.55 (60.114) & 3.70 (0.082) & 1288.29 (4.250) & 5.69 (0.033) & 2.74 (0.005) \\
  $\Mb$ & 2.06 (0.001) & 2382.76 (50.235) & 3.42 (0.074) & 2244.12 (21.266) & 50.37 (1.096) & 3.91 (0.032) \\
  $\Mc$ & 2.07 (0.001) & 2365.50 (44.463) & 3.50 (0.064) & 1342.36 (4.859) & 3.87 (0.011) & 2.61 (0.005) \\
  $\Md$ & 2.06 (0.001) & 2377.28 (42.434) & 3.41 (0.062) & 1387.39 (4.302) & 7.19 (0.038) & 2.86 (0.010) \\
   \hline
\end{tabular}
\end{table}

\begin{figure}[t]
        \begin{center}
                \includegraphics[width=0.99\linewidth]{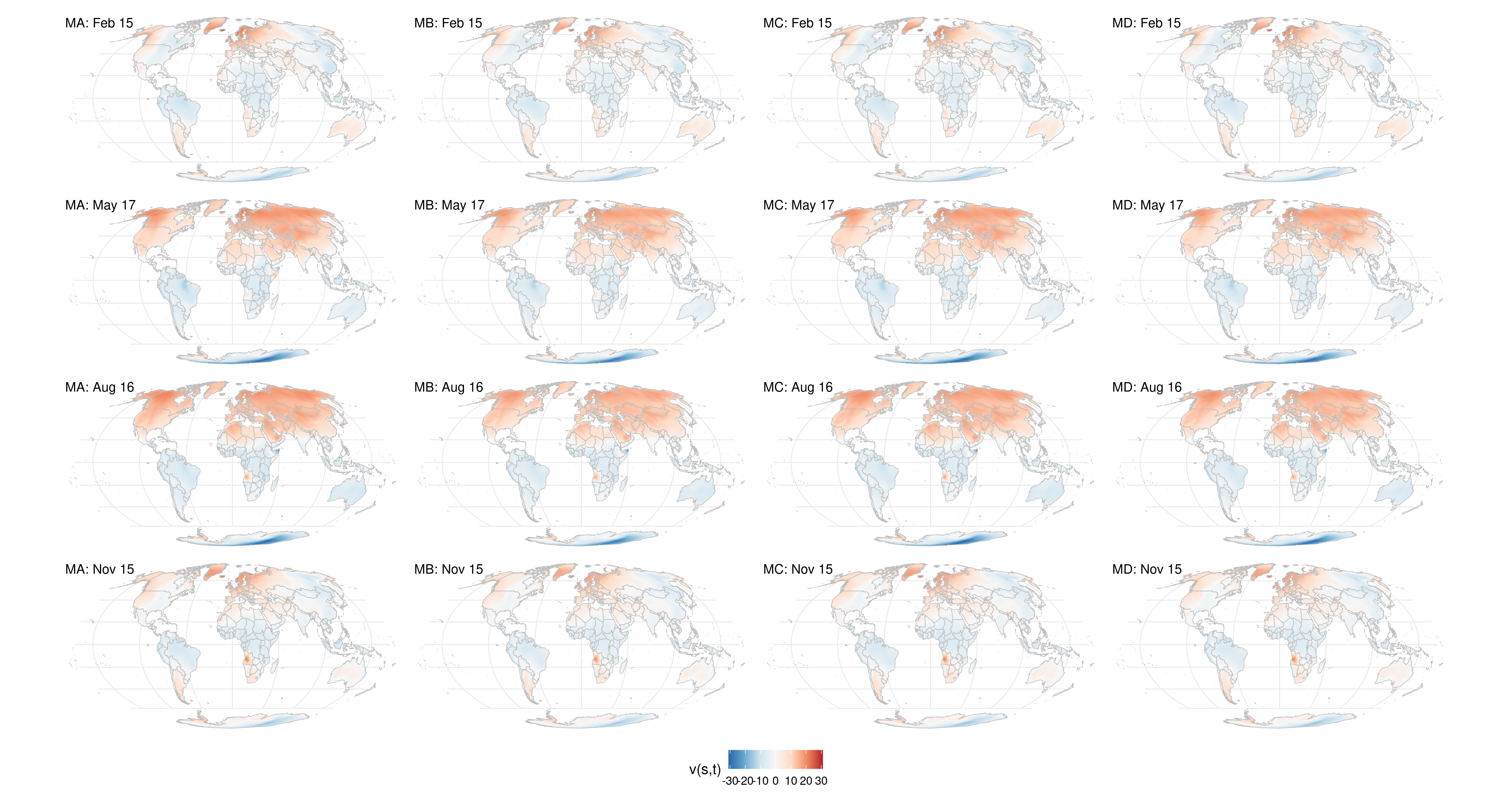}
        \end{center}
        \vspace{-0.5cm}
        \caption{
            The posterior mean for $v(\s,t)$, at some time points, for each model.
	    From left to right, models $\Ma$, $\Mb$, $\Mc$, and $\Md$.
            From top to bottom, the time points are 46, 137, 228, and 319,
            corresponding to the day of year in 2022 as labelled at the top right of each plot.
            }
        \label{figv4fields}
\end{figure}

\begin{figure}[t]
        \begin{center}
                \includegraphics[width=0.99\linewidth]{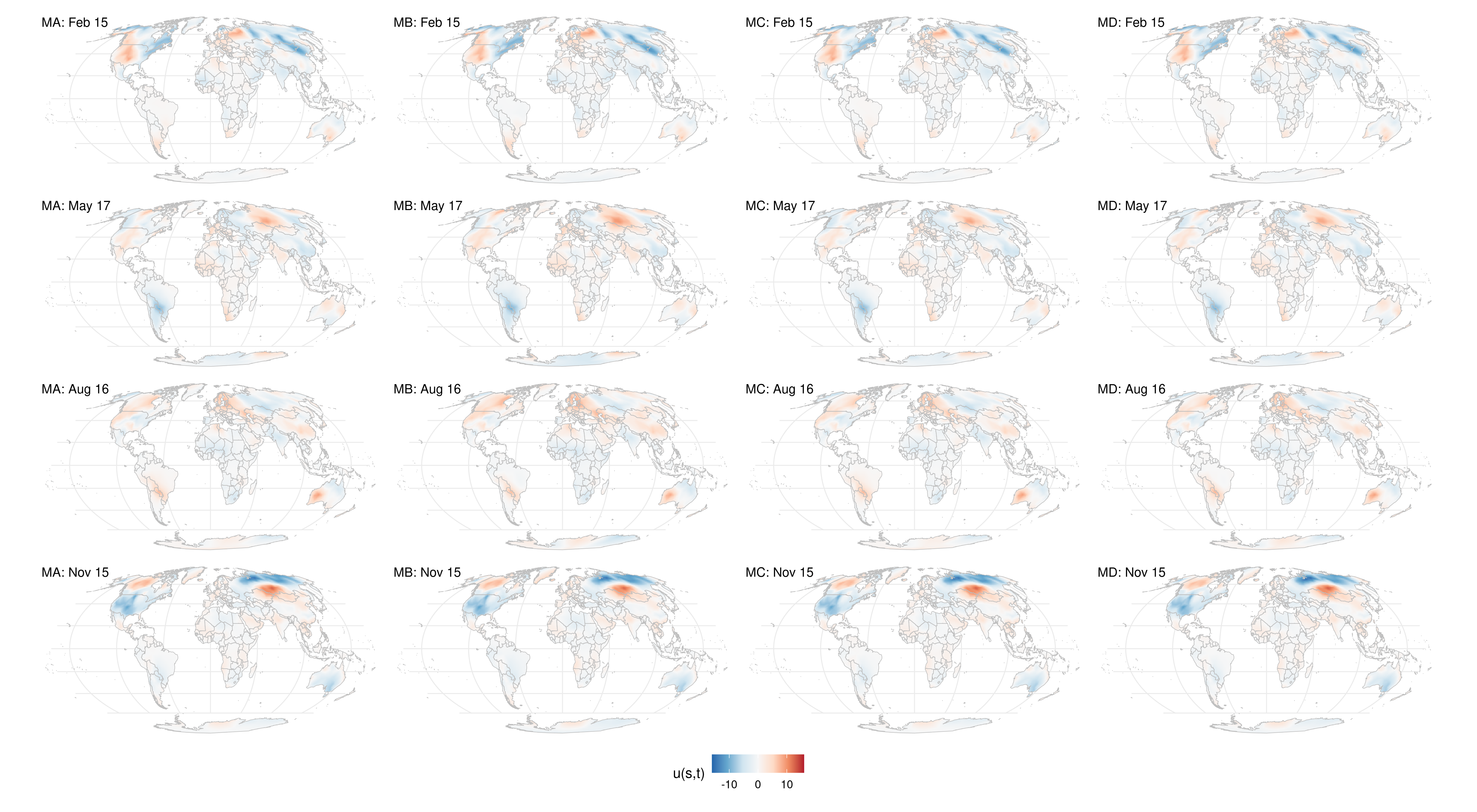}
        \end{center}
            \vspace{-0.5cm}
        \caption{
            The posterior mean for $u(\s,t)$, at some time points, for each model.
	    From left to right, models $\Ma$, $\Mb$, $\Mc$, and $\Md$.
            From top to bottom, the time points are 46, 137, 228, and 319,
            corresponding to the day of year in 2022 as labelled at the top right of each plot.
            }
        \label{figst4fields}
\end{figure}

For $\Mb$ (using model~$B$ for $u(\s,t)$), the posterior mean for
$\mu$ is $34.74$ and for $\alpha$ is $-4.70$.
The posterior mean of $b(\s,t)$ for all four models are shown in
Figure~\ref{figseasonallatitude}, which displays the temperature over time and latitude.
The seasonal pattern is clear, with summer and winter temperatures in the two hemispheres
standing out, with in particular lower temperatures (blue) in each hemisphere's respective
winter.
The model range/variance parameter estimates are summarised in Table~\ref{tabhysummar}.
For the posterior mean, we have $\pE(r_s\mid\mv{y})=2244.12$~km and
$\pE(r_t\mid\mv{y}) = 50.37$~days for model $\Mb$,
and smaller for the other models.
These values can be interpreted through Figure~\ref{fig:covariancesA}.
The posterior mean for the spatio-temporal field $v(\s,t)$ for some days
in 2022 is shown in Figure~\ref{figv4fields}.
This term captures temporal slowly varying spatial variation
from the overall mean, elevation effect and the basic seasonal latitude parts of the model.
The posterior mean for the spatio-temporal field $u(\s,t)$ for some days
in 2022 is shown in Figure~\ref{figst4fields}. This term captures the remaining
spatio-temporal variation of the temperature field around the other parts of the model.

\subsubsection{Forecast evaluation}
\label{sec:forecast-evaluation}
As was already apparent from the diagnostic scores in Table~\ref{tab1gt},
despite the temporal range parameters being different for the four models, particularly form model $\Mb$,
they are nearly indistinguishable with respect to direct and leave-one-out prediction distributions.
Since the space-time non-separability effect is unclear in the leave-one-out
setting, we extend the assessment by computing multi-horizon temporal predictions.
We used the first 14 days of the data from each month to predict the following 7 days.
These forecasts were done while keeping the covariance parameters and the long term spatio-temporal components $b(\s,t)$ and $v(\s,t)$ fixed to their posterior modes from the full joint model estimates, so that only
the short-term spatio-temporal field $u(\s,t)$ was reestimated for each scenario.
This generated forecasts for each model for 12 different weather and seasonal conditions over the year.

\begin{figure}[t]
        \begin{center}
                \includegraphics[width=0.89\linewidth]{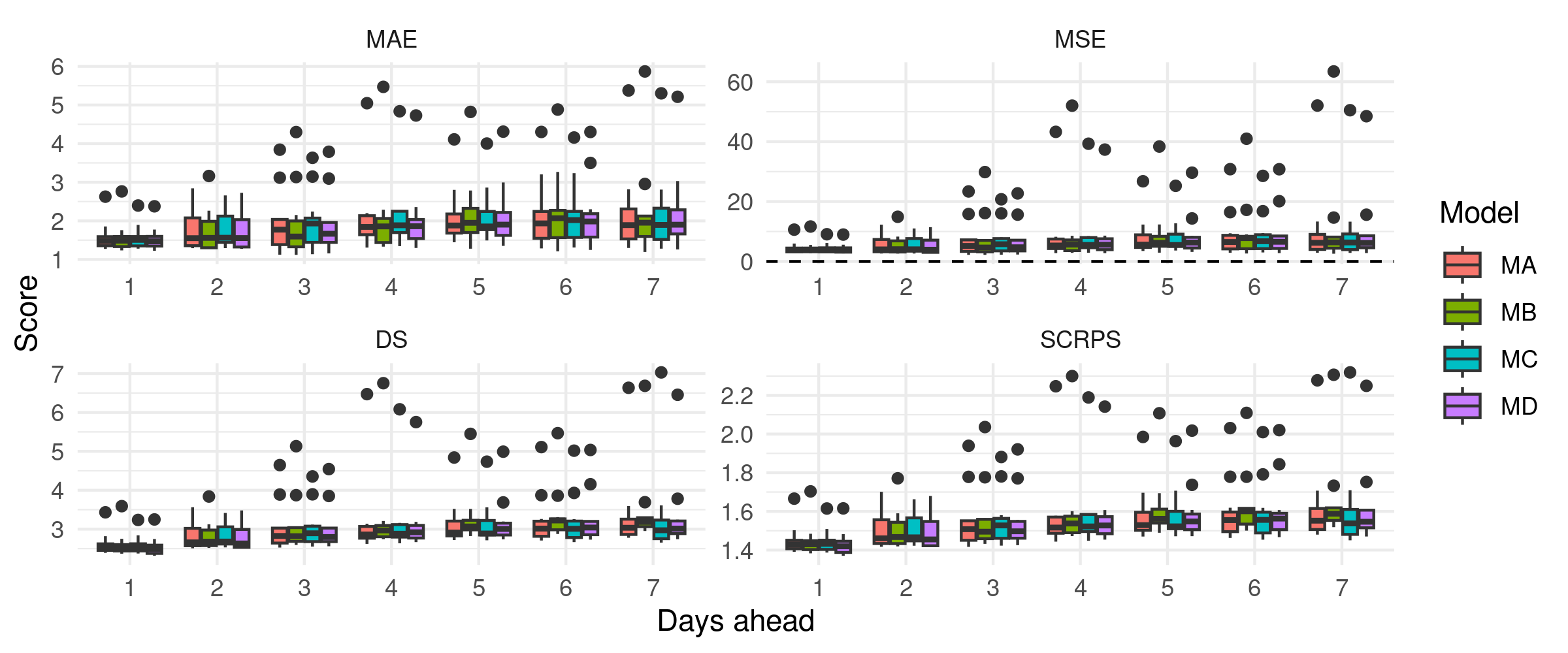}
                \includegraphics[width=0.89\linewidth]{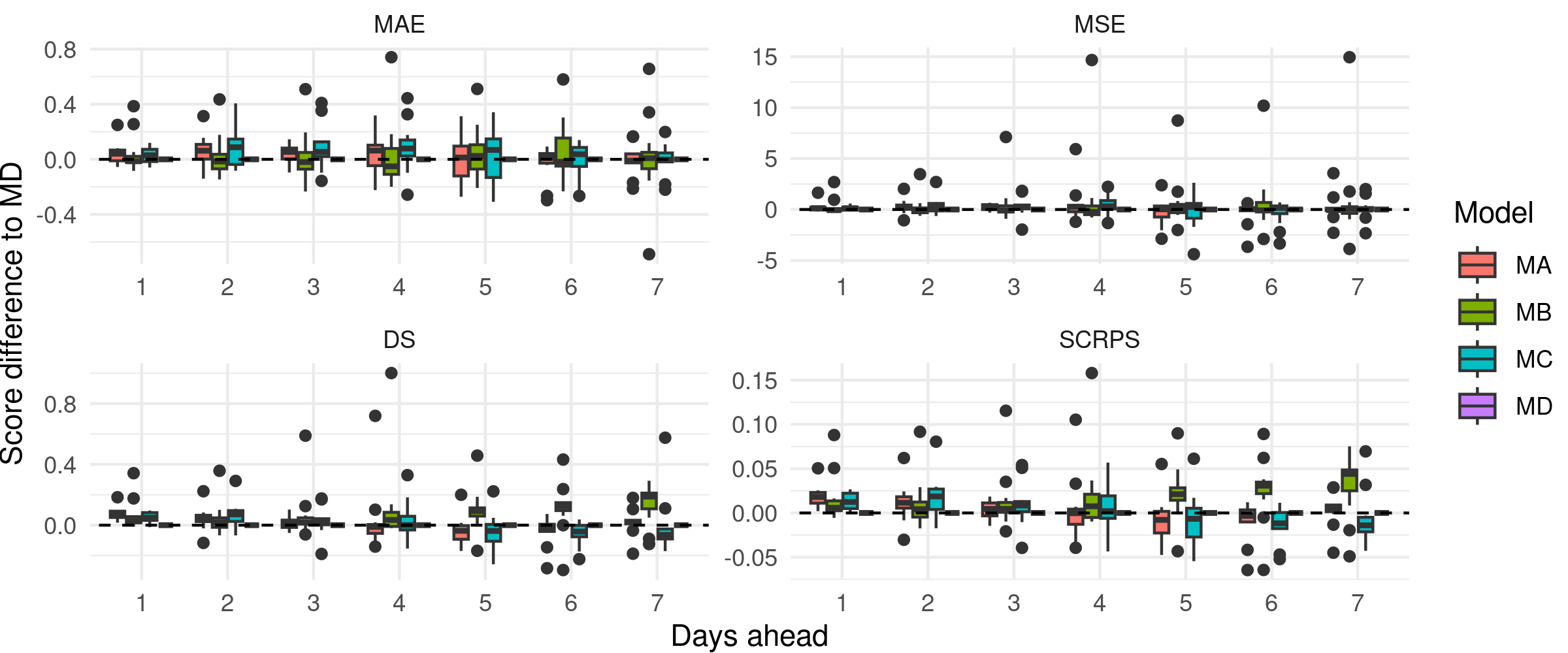}
        \end{center}
            \vspace{-0.5cm}
        \caption{
            Top: Multi-horizon (1--7 days) and multi scenario (one for each month of year 2022)
            forecast scores for predicting one week ahead. Lower scores indicate a better forecast.
            Bottom: The differences in scores compared with model $\Md$.
            }
        \label{fig:multi-horizon}
\end{figure}

Figure~\ref{fig:multi-horizon}(top) shows the mean absolute error (MAE), mean squared error (MSE),
mean Dawid-Sebastiani \citep[DS, equivalent to log-score for Gaussian predictions, see][]{gneiting2005}, and mean SCRPS summarized for each prediction horizon (1--7 days) for each of the 12 scenarios.
Figure~\ref{fig:multi-horizon}(bottom) shows the difference between the scores for each model to those of model $\Md$, to more clearly highlight the differences between the models.
The prediction errors all exhibit increasing variability for longer forecast horizons, as well as a generally increasing trend, that mostly levels off around 6 days, which is compatible with the estimated temporal correlation length parameter $r_t$ for models $\Ma$, $\Mc$, and $\Md$.
For 1-day ahead forecasts, model $\Md$ achieved the lowest scores, and it appears
more stable than the other models for longer forecast horizons. Model $\Mb$
has large score variability, and is
doing worse than the other three models for long forecast horizons, in particular for the scores that take forecast uncertainty into account.
For more details see
Appendix~\ref{sec:app-details}, where one can see that the scores are generally worse
in the start and end of the year, indicating an unmodelled aspect of seasonality,
e.g.\ in weather variability.

\section{Discussion}\label{sec:discussion}

We have developed a spatio-temporal extension of
the Gaussian Mat\' ern fields based on
a fractional and stochastic version of
the physical diffusion equation considered
by \citet{whittle1954,whittle1963stochastic}.
We named the new family the Diffusion-based
Extension of the Mat\'ern Field (DEMF),
and showed that it
has several useful properties:
The spatial marginals are Gaussian Mat\'ern fields;
the family contains
Markovian diffusion processes with clear physical interpretations; and
we can control the smoothness in space and in time,
the degree of non-separability,
and interpret all the parameters.
The family can also be extended to
non-stationary models and be defined on curved manifolds.

The DEMF family contains several important subfamilies;
1) Separable models, 2) Markov models, 3) Partially separable models,
4) a fully non-separable subfamily of the \citet{stein2005space} family,
and 5) spatially non-stationary model dynamics.
This provides a rich outset for studying the
practical and methodological impacts these assumptions
have.

An important special case
in the DEMF family is the  DEMF(1,2,1) model, which in two-dimensional space, is the closest
stochastic process analogue to the
diffusion equation
(see \eqref{eq:whittlespacetime}),
and hence a natural default choice
for spatio-temporal model components.
The non-separable DEMF(1,2,1) model
has the same smoothness in space and in time
as the separable DEMF(1,0,2) model,
which has a covariance function that is a Kronecker product of
a Mat\'ern covariance in space and an exponential covariance in time.
Of particular interest is also the non-separable DEMF(2,2,0) model which can be viewed as an iterated diffusion model.

Although the proposed model family includes non-separable models, which
in itself might be desirable from considerations about covariance properties,
another view-point is that the non-separability here arises as a direct and natural consequence of
the physics-inspired dynamical diffusion construction.
Most importantly, the results shed light on which types of non-separability
would occur naturally under certain assumptions on the spatio-temporal dynamics
and properties of the driving noise process.
Although there are strong arguments in the literature
against using a separable model,
the space of non-separable models is vastly larger than the space of separable models.
Hence we need to consider which types of non-separable models
are more, and which are less, appropriate than the separable alternatives.
As illustrated by the practical example in Section\ref{sec:globaltemp}, it is important
to assess models in a context relevant to the intended use case. In particular,
non-separability is unlikely to make a difference for space-time interpolation,
as assessed by e.g.\ leave-one-out cross-validation, but can make a difference
in full space-time forecasting settings.

It is natural to view the model class as an example of building models via building blocks
with precision operator space-time separability.
The most basic form of separability is \emph{functional separability}, where
a spatial and temporal processes are added or multiplied, which can be viewed as having $S+T$ degrees of freedom, where
$S$ and $T$ are the spatial and temporal effective dimensions of the functions.
The next form is \emph{covariance separability}, where the model is formed from
a sum of covariances (giving the same as functional separability) or a product
of covariances, where the latter gives $S\cdot T$ degrees of freedom. These
covariance product models are covariance separable but functionally non-separable.
For precision models, plain products are equivalent to covariance separable models,
but sums of precision products give covariance non-separability.
In both the covariance and precision cases, non-stationarity in the spatial and
temporal operators can be introduced, as long as the operator separability is kept.
This distinguishes this type of non-separability from fully non-separable non-stationary models
that cannot be written as precision sums and products.  The key is to retain
commutativity between the spatial and temporal operators within each product:
$(\mb{Q}_t \otimes \mb{I}_s)(\mb{I}_t \otimes \mb{Q}_s)=(\mb{I}_t \otimes \mb{Q}_s)(\mb{Q}_t \otimes \mb{I}_s)=\mb{Q}_t \otimes \mb{Q}_s$.

With the GMRF representation presented herein,
the computational costs of the separable and non-separable models are similar, as the sparsity structure
of posterior precisions, given irregularly spaced observations in generalised latent Gaussian models,
is only marginally affected by the non-separability, and can even be more sparse
in the non-separable cases; the separable precision neighbourhood structures are space-time prisms,
whereas the non-separable neighbourhood structures are double-cones.
Together with interpretable parameters,
this makes the non-separable models
as practically accessible as the separable models.
In the supplementary materials
we provide an implementation with examples in \texttt{R-INLA}.
%This implementation includes model specification in R, which can
%be extracted and used in other statistical software,
%e.g.\ MGCV \citep{wood2011fast} and
%TMB \citep{kristensen2016tmb}.

In this paper we mainly focused on stationary fields,
but also showed how very little in the theory and computational construction
changes for models with
curved manifolds or spatially non-stationary operators, as already discussed by
\citet{lindgren2011explicit}.
Although the initial practical implementation only covers a subset of the general model class,
we believe that the general results can and will be applied in more general contexts
in the future.

\section{Supplementary materials}

The examples were computed with the \texttt{INLAspacetime} package, using
the \texttt{cgeneric} method from the \texttt{R-INLA} software for
computationally efficiency, via the \texttt{inlabru} interface
\citet{bachl_inlabru_2019}. The code for the example can be found in
the supplementary material. See also \citet{van2019new} for a similar
example.

Code for the figures and examples is available
at \url{https://github.com/finnlindgren/spacetime-paper-code}, and the
\texttt{INLAspacetime} \textsf{R} package
(\url{https://github.com/eliaskrainski/INLAspacetime}) implements a subset of
the models.

\section{Acknowledgements}
As part of the EUSTACE project, Finn Lindgren received funding from the European
Union's Horizon 2020 Programme for Research and Innovation, under
Grant Agreement no. 640171.

\bibliography{spacetime_bib,mybib}

\appendix

\newpage
\section{Almost sure sample path continuity}
\label{sec:path-continuity}

We start by rephrasing the main
theorem of Section~9.3 of \citet{cramer1967book}, and giving a formal definition
of the smoothness index.

\begin{definition}[Cram\'er and Leadbetter, Section 2.5, generalised]
A stochastic process $x(t)$ on some domain $\cD$, is \emph{equivalent} to another process $y(t)$ on $\cD$, if for each fixed $t\in\cD$, $x(t)=y(t)$, with probability one. This means that $x$ differs from $y$ on at most a set with measure zero, and that they have the same finite dimensional distributions.
\end{definition}
This technical definition allows us to view equivalent processes as an equivalence class that encapsulates some of the finer details of probabilistic measure theory for sample path continuity of stochastic processes.

\begin{theorem}[Cram\'er and Leadbetter, Section 9.3]
\label{thm:cramer}
Let $S^*(\omega)$ be the spectral measure of a stationary Gaussian process $x(t)$ on $t\in\mathbb{R}$,
and let
\begin{align*}
I_{a,b} &= \int_0^\infty \omega^{2a} \left[\log(1+\omega)\right]^b \md S^*(\omega)
\end{align*}
for $a,b \geq 0$.
For spectral measures that admit a spectral density $S(\omega)$, replace $\mmd S^*(\omega)$ in $I_{a,b}$ with $S(\omega)\md\omega$.
\begin{enumerate}
\item If $I_{a,b}<\infty$ for some $b>3$ and some $a$ in the range $[k,k+1)$ for some $k\in\mathbb{N}$, then $x(t)$ is equivalent to a process $y(t)$ that has a
continuous sample derivative of order $k$, with probability one.
\item If $I_{a,1}<\infty$ for some $a$ in the range $(k,k+1]$ for some $k\in\mathbb{N}$, then $x(t)$ is equivalent to a process $y(t)$ whose sample derivative of order $k$ is Hölder continuous with exponent $a-k\in(0,1]$, with probability one.
\end{enumerate}
For the case $k=0$, the sample derivative of order zero refers to the sample
path of the process itself.
\end{theorem}
\begin{proof}
The results follow directly from the main theorem of Section 9.3 of
\citet{cramer1967book}.
\qqed
\end{proof}

Results from \citet{scheuerer_regularity_2010} show that under a similar condition
for $d$-dimensional domains,
\begin{align*}
\int_{\RR^d} \|\mb{\omega}\|^{2a} \left[\log(1+\|\mb{\omega}\|)\right]^b \md S^*(\mb{\omega})
&< \infty,
\end{align*}
for all $a<\nu$ and some $b>1$,
the sample paths on $\RR^d$ belong to any Sobolev space $W^{a,2}$ of order $a<\nu$,
on any bounded subdomain, with probability one.
For isotropic spectra, this translates to $I_{a,b}<\infty$ for all $a<\nu$ and some $b>1$,
when applied to the one-dimensional marginal spectra.

The integral criteria above motivate the following characterisation of the \emph{smoothness index} $\nu$,
in particular when applied to models with power law spectral density tails.
\begin{definition}
The smoothness index $\nu$ of a stationary Gaussian process $x(t)$, $t\in\mathbb{R}$,
is
$
\nu=\sup_{a} \{a; I_{a,1} < \infty\},
$
where $I_{a,1}$ is defined as in Theorem~\ref{thm:cramer}.
\end{definition}

% !TEX root = spacetime_jrssb.tex

\section{Numerical evaluation of covariances}\label{app:spec2cov}

When spatio-temporal spectral density is available in closed format on $\RR^d\times\RR$,
the covariance function can be obtained to close numerical accuracy using fast Fourier transformation (FFT).
In order to reduce the memory requirements for isotropic models on high-dimensional spatial domains, the marginal space-time spectrum along a single spatial dimension can be evaluated
first. For general models, evaluating spatial FFT transformations for each time lag further reduces the memory footprint if only some of the covariances are stored.

The idea is construct the folded spectrum resulting from spatial/temporal discretisation, and then discretise it onto a finite regular lattice. The resulting integral approximations can be evaluated with standard FFT implementations, and the numerical approximation error in the covariance evaluation is determined by the the frequency resolution and smoothness of the spectral density.  The brief theory behind the construction presented below is based on \citet{LindgrenStocProc}.

\subsection{Spectral folding}

The exact spectral representation of the covariance evaluated on a discrete infinite lattice can be derived from the continuous domain representation. For simplicity, assume the same lattice spacing $h$ in each direction. A stationary covariance function $R(\s)$ evaluated at lattice points $\mv{j} h$, $\mv{j}\in\ZZ^d$ is given by
\begin{align}
R(\s) &= \int_{\RR^d} \exp(i\mv{\omega}\cdot\s) S(\mv{\omega}) \md\mv{\omega}, \notag \\
R(\mv{j}h) &= \int_{\RR^d} \exp(i\mv{\omega}\cdot\mv{j} h) S(\mv{\omega}) \md\mv{\omega} \notag \\
 &= \int_{[-\pi/h,\pi/h)^d} \sum_{\mv{k}\in\ZZ^d} \exp(i(\mv{\omega} + 2\pi \mv{k}/h) \cdot \mv{j}h)
 S(\mv{\omega} + 2\pi \mv{k}/h) \md\mv{\omega} \notag
 \\
 &= \int_{[-\pi/h,\pi/h)^d} \exp(i \mv{\omega} \cdot \mv{j} h)
 \wt{S}(\mv{\omega})
 \md\mv{\omega},
 \label{eq:folded-spectrum}
\end{align}
where
\begin{align*}
\wt{S}(\mv{\omega}) &=\sum_{\mv{k}\in\ZZ^d} S(\mv{\omega} + 2\pi \mv{k}/h),
\quad \mv{\omega}\in [-\pi/h,\pi/h)^d.
\end{align*}
If instead the spatial discretisation should be interpreted as the \emph{cell averages} (which is the more usual case for PDE discretisations and e.g.\ satellite data, rather than pointwise values), the spectrum is altered by a multiplicative frequency
filter with a squared sinc function:
\begin{align*}
\wt{S}(\mv{\omega}) &=\sum_{\mv{k}\in\ZZ^d} S(\mv{\omega} + 2\pi \mv{k}/h)
\prod_{l=1}^d
\left\{
\frac{
  \sin[(\omega_l + 2\pi k_l/h)/2]
}{
(\omega_l + 2\pi k_l/h)/2
}
\right\}^2,
\quad \mv{\omega}\in [-\pi/h,\pi/h)^d .
\end{align*}

\subsection{Discrete Fourier transformation}

To approximate the integral in \eqref{eq:folded-spectrum} with FFT, choose
a positive integer $M$ .
This gives a numerical integration approximation
\begin{align}
\wh{R}(\mv{j} h) &=
\left(\frac{\pi}{hM}\right)^d\sum_{\mv{k}\in[-M,M)^d}
\exp\left(i \mv{k} \cdot \mv{j} \frac{2\pi}{2M}\right)
 \wt{S}\left(\mv{k}\frac{\pi}{hM}\right)
 , \quad \mv{j}\in [-M,M)^d ,
 \label{eq:spec2cov}
\end{align}
which is of the form that can be evaluated using FFT.

\subsubsection{Sampling}

With the above theory, sampling from the model can be expressed as an integral with respect to continuous domain complex valued white noise process, $\mmd Z(\mv{\omega})$, with conjugate symmetry:
\begin{align*}
x(\mv{j}h)
 &= \int_{[-\pi/h,\pi/h)^d} \exp(i \mv{\omega} \cdot \mv{j} h)
\wt{S}(\mv{\omega})^{1/2} \md Z(\mv{\omega}), \quad \mv{j}\in \ZZ^d,
\end{align*}
where $\overline{\mmd Z(-\mv{\omega})}=\mmd Z(\mv{\omega})$, $\textsf{Cov}(\mmd Z(\mv{\omega}),\mmd Z(\mv{\omega}'))= \delta(\mv{\omega}-\mv{\omega}')\md\mv{\omega}$. This can be discretised with a lattice of frequencies in much the same way as for computing the covariance function, with noise variances equal to the cell area/volume $\left(\frac{\pi}{hM}\right)^d$ of each frequency lattice point. When the outer pairwise opposing cells are discretised, the combined complex noise contributions are real, and should be assigned to the $-M$ indices, which ensures that the resulting field has no non-zero imaginary components.

\section{Spherical harmonics}
\label{sec:spherical-harmonics}

\newcommand{\F}{\mathcal{F}}
\newcommand{\Fi}{\mathcal{F}^{-1}}

\subsection{Definition and standard properties}
In $\RR^2$, the harmonic functions, sine and cosine, play an important
role as basis functions in spectral representations of functions and
random fields.  On the sphere, this role is instead taken by the
\emph{spherical harmonics}.  This section presents the basic results needed
for spectral representation theory for stationary processes on the sphere.

\begin{definition}
The \emph{spherical harmonic} $Y_{k,m}(\mv{u})$,
$\mv{u}=\mat{u_1,u_2,u_3}^\top\in \sphere^2\subset\RR^3$, of \emph{order}
$k=0,1,2,\ldots$ and \emph{mode} $m=-k,\ldots,k$ is defined by
\begin{align*}
Y_{k,m}(\mv{u}) &=
\sqrt{(2k+1)\cdot\frac{(k-|m|)!}{(k+|m|)!}}
\cdot
\begin{cases}
\sqrt{2} \sin(m\phi) P_{k,-m}(\cos\theta) & -k\leq m <0,\\
P_{k,0}(\cos\theta) & m=0,\\
\sqrt{2} \cos(m\phi) P_{k,m}(\cos\theta) & 0 < m \leq k,
\end{cases}
\end{align*}
where $\phi$ is the longitude and $\theta=\arccos(u_3)$ is the
colatitude, and $P_{k,|m|}(u_3)$ are \emph{associated Legendre
  functions} ($P_{k,0}(u_3)$ are Legendre polynomials).  Note that
$\sin{\phi}=u_2/\sqrt{u_1^2+u_2^2}$, $\cos{\phi}=u_1/\sqrt{u_1^2+u_2^2}$, and
$\cos\theta=u_3$.
\end{definition}
Standard property results for spherical harmonics,
following \citet{wahba_spline_1981}, building the basis of spherical Fourier theory:
\begin{enumerate}
\item The spherical harmonics form an orthogonal basis for functions
  on the unit sphere, $\sphere^2$:
\begin{align*}
\< Y_{k,m}, Y_{k',m'}\>_{\sphere^2} &=
\begin{cases}
4\pi, & k'=k, m'=m,\\
0, & \text{otherwise.}
\end{cases}
\end{align*}
\item The addition formula for spherical harmonics is
\begin{align*}
\sum_{m=-k}^k Y_{k,m}(\mv{u})Y_{k,m}(\mv{v}) &=
(2k+1) P_{k,0}(\mv{u}^\top\mv{v}) .
\end{align*}
\item The spherical harmonics are eigenfunctions to the Laplacian on
  $\sphere^2$,
\begin{align*}
\Delta Y_{k,m}(\mv{u}) &= -k(k+1)Y_{k,m}(\mv{u}) .
\end{align*}
\item Let $\phi(\mv{u})$ be a square-integrable function on $\sphere^2$.
Then $\phi(\mv{u})$ has series expansion
\begin{align*}
\phi(\mv{u}) &=
(\Fi\wh{\phi})(\mv{u}) =
\sum_{k=0}^\infty\sum_{m=-k}^{k} \wh{\phi}(k,m) Y_{k,m}(\mv{u}),
\end{align*}
with \emph{Fourier Bessel} coefficients
$\wh{\phi}(k,m) =
(\F\phi)(k,m) =
\frac{1}{4\pi}
\< {\phi}(\mv{u}), Y_{k,m}(\mv{u}) \>_{\sphere^2(\mmd\mv{u})}$.
Also,
$\< \phi, 1 \>_{\sphere^2} = 4\pi \wh{\phi}(0,0)$ and 
$\< \phi, \phi \>_{\sphere^2} = 4\pi \sum_{k,m} \wh{\phi}(k,m)^2$.
\end{enumerate}

\subsection{Spherical variance approximation}
\label{sec:spherical-variance-approximation}

Define
\begin{align*}
F_{a,b} &= \sum_{k=a}^b \frac{2k+1}{4\pi [\gamma_s^2 + k(k+1)]^\alpha} ,
\end{align*}
so that $F_{0,\infty}$ gives the variance in \eqref{eq:spherical-variance}.
With
\begin{align*}
I_{a,b}&=\int_a^b \frac{2x+1}{4\pi[\gamma_s^2+x(x+1)]^\alpha} \md x \\
&=
\frac{1}{4\pi(\alpha-1)}\left(
\frac{1}{[\gamma_s^2+a(a+1)]^{\alpha-1}}
-
\frac{1}{[\gamma_s^2+b(b+1)]^{\alpha-1}}
\right) ,
\end{align*}
choose $K$ so that the terms in the sum \eqref{eq:spherical-variance}
are decreasing for $k\geq K$. This holds for any $K \geq K_0$,
where $K_0=0$ if $\gamma_s\leq 1/2$, and $K_0=\ceil{\sqrt{\frac{\gamma_s^2-1/4}{2\alpha-1}}-\frac{1}{2}}$ for $\gamma_s > 1/2$. Then the full sum $F_{0,\infty}$ can be bounded by a partial sum $F_{0,K}$ and tail integrals:
\begin{align*}
F_{0,K} + I_{K+1,\infty} &\leq F_{0,\infty} \leq F_{0,K} + I_{K,\infty} .
\end{align*}
Tighter bounds can in principle be obtained for the approximation $F_{0,\infty}\approx F_{0,K} + I_{K+1/2,\infty}$.  Let $f_x$ denote the integrand for $I_{a,b}$. Then a second order Taylor expansion
around each $x=k$ gives the error bound
\begin{align*}
\left|F_{0,K} + I_{K+1/2,\infty} - F_{0,\infty}\right|
&=
\left|I_{K+1/2,\infty} - F_{K+1,\infty}\right|
\leq
\frac{1}{24}
\sum_{k=K+1}^\infty
\sup_{x\in(k-1/2,k+1/2)} \left|f^{''}_{x}\right| .
\end{align*}
It may be possible to construct a bound for this series using another integral bound,
but the practical utility of doing so is unclear.

\begin{comment}
%% TODO: complete this, using notes from elsewhere
To find a sufficiently large $K$ for a given relative error $\epsilon$, we
use the bounds and solve for $K$:
\begin{align*}
\frac{|F_{0,K}+I_{K+1/2,\infty}-F_{0,\infty}|}{F_{0,\infty}}
&\leq
\frac{|I_{K,\infty} - I_{K+1,\infty}|}{F_{0,K_0} + I_{K_0+1,\infty}} \leq
\frac{|I_{K,K+1}|}{F_{0,K_0} + I_{K_0+1,\infty}} \leq \epsilon ,
\end{align*}
for sufficiently large $K$.
The numerator becomes
\begin{align*}
I_{K,K+1} &\leq f_K =
\frac{2K+1}{4\pi[\gamma_s^2+K(K+1)]^{\alpha}}
\end{align*}
\end{comment}

% !TEX root = spacetime_jrssb.tex

\section{Collected proofs}\label{app:proofs}

\subsection{Proof of Proposition~\ref{prop:spatial_cov}}  \label{proof:spatial_covar}
The covariance function for spatial lag $\s=\s_2-\s_1$ and temporal lag $t$
can be written as a nested integral,
\begin{align*}
\pCov[u(\mb{0},0),u(\s,t)] &=
\int_{\RR^d}
\int_{\RR}\exp[i (\s\cdot\oms + t\omt)] S_u(\oms,\omt)
\md\omt \md\oms
\\&=
\int_{\RR^d}
\exp(i \s\cdot\oms)
\left\{
\int_{\RR} \exp(i t\omt) S_u(\oms,\omt)
\md\omt
\right\}
\md\oms .
\\&=
\int_{\RR^d}
\exp(i \s\cdot\oms)
S_u(\oms;t)
\md\oms ,
\end{align*}
where the inner integral $S_u(\oms;t)$ is the marginal spatial cross-spectrum for time lag $t$.

Let $\lambda=\gamma_s^2+\|\oms\|^2$ and $\kappa^2=\lambda^{\alpha_s}/\gamma_t^2$.
Then, integrating over $\omt$, we get
\begin{align*}
S_u(\oms;t) &= \frac{1}{(2\pi)^d\gamma_e^2\lambda^{\alpha_e}\gamma_t^{2\alpha_t}}
\int_{\RR} \frac{\e^{i t \omt}}{2\pi (\omt^2 + \lambda^{\alpha_s}/\gamma_t^2)^{\alpha_t}} \md \omt \nonumber \\
&=
\frac{1}{(2\pi)^d\gamma_e^2\lambda^{\alpha_e}\gamma_t^{2\alpha_t}}
\frac{C_{\RR,\alpha_t}}{\kappa^{2(\alpha_t-1/2)}}
R^M_{\alpha_t-1/2}(\kappa t)
\nonumber \\
&=
\frac{C_{\RR,\alpha_t}}{\gamma_e^2 \gamma_t}
\frac{1}{(2\pi)^d(\gamma_s^2+\|\oms\|^2)^{\alpha}}
R^M_{\alpha_t-1/2}\left\{t \sqrt{\gamma_s^2+\|\oms\|^2} / \gamma_t\right\}
,
\end{align*}
where $R^M_{\nu}(t)$ is the standard Mat\'ern correlation with smoothness $\nu$,
defined in \eqref{eq:matern-cov}, and the additional scaling was given in
 \citet{lindgren2011explicit}.
For $t=0$, the temporal constribution factor is $1$, and we recognize the
resulting expression as the spectral density corresponding to a spatial Mat\'ern
covariance function with range parameter $\gamma_s$ and
smoothness parameter $\nu_s = \alpha - d/2$, and marginal variance equal to the
sought value $\sigma^2$ in the proposition. We then also know that the marginal spectrum for $t=0$ in any single spatial dimension is proportional to $(\gamma_s^2+\omega^2)^{-\nu_s+1/2}$, which shows that the conditions on $a$ in Theorem~\ref{thm:cramer} are fulfilled if and only if $a<\nu_s$, so $\nu_s$ is the smoothness index.

\subsection{Proof of Proposition \ref{prop:tempsmooth}} \label{proof:tempsmooth}

Let $\nu_t$ be the smoothness index for the marginal temporal process $u(\s,t)$.
We need to identify for which values of $a$ the integral $I_{a,1}=\int_0^\infty \omt^{2a}\log(1+\omt) S_u(\omt) \md\omt$ in Theorem~\ref{thm:cramer} is finite.
We start by integrating out the spatial spectral dimensions and reparameterising the resulting integral:
\begin{align}
S_u(\omt)
&\propto  \int_{\RR^d}
[\gamma_t^2\omt^{2}+ (\gamma_s^2+\| \mv{\omega}_s \|^2)^{\alpha_s}]^{-\alpha_t}
(\gamma_s^2 + \| \mv{\omega}_s \|^2)^{-\alpha_e} \md\mv{\omega}_s  \nonumber\\
&\propto  \int_{0}^{\infty} r^{d-1}
[\gamma_t^2\omt^{2}+
(\gamma_s^2+r^2)^{\alpha_s}]^{-\alpha_t}
(\gamma_s^2 + r^2)^{-\alpha_e} \md r \nonumber\\
&\propto \int_{0}^{\infty}
v^{(d-2)/2}(1+v)^{-\alpha_e}(\wt{\omega}_t^2 + (1+v)^{\alpha_s})^{-\alpha_t}\md v
\label{eq:initial-integration-steps}
\end{align}
where we in the second step changed to polar coordinates and in the third set $v = r^2/\gamma_s^2$ and $\wt{\omega}_t = \omt\gamma_t/\gamma_s^{\alpha_s}$.
The integral \eqref{eq:initial-integration-steps} is finite for all $\wt{\omega}_t$ when $\alpha_e+\alpha_s\alpha_t > d/2$.
Assuming $a<\nu_t$, we can then write the integral in the smoothness criterion as
\begin{align*}
I_{a,1} &=
\int_0^\infty \omt^{2a}\log(1+\omt) S_u(\omt) \md\omt \\
&=
C_0\int_0^\infty \wt{\omega}_t^{2a}\log\left(1+\frac{\wt{\omega}_t\gamma_s^{\alpha_s}}{\gamma_t}\right)
\int_{0}^{\infty}
v^{(d-2)/2}(1+v)^{-\alpha_e}(\wt{\omega}_t^2 + (1+v)^{\alpha_s})^{-\alpha_t}\md v
\md\wt{\omega}_t
\end{align*}
for some constant $C_0$.
Let $\epsilon > 0$ such that $a+\epsilon<\nu_t$. Then $\log\left(1+\frac{\wt{\omega}_t\gamma_s^{\alpha_s}}{\gamma_t}\right) \leq C_\epsilon\wt{\omega}_t^{2\epsilon}$ for all $\wt{\omega}_t>0$ for some $C_\epsilon>0$. We can then bound $I_{a,1}$ and change the order of integration since the integrands are positive:
\begin{align*}
I_{a,1} &\leq
C_0C_\epsilon\int_0^\infty \wt{\omega}_t^{2(a+\epsilon)}
\int_{0}^{\infty}
v^{(d-2)/2}(1+v)^{-\alpha_e}(\wt{\omega}_t^2 + (1+v)^{\alpha_s})^{-\alpha_t}\md v
\md\wt{\omega}_t \\
&=
C_0C_\epsilon
\int_{0}^{\infty}
v^{(d-2)/2}(1+v)^{-\alpha_e}
\int_0^\infty
\frac{\wt{\omega}_t^{2(a+\epsilon)}}{
(\wt{\omega}_t^2 + (1+v)^{\alpha_s})^{\alpha_t}
}
\md\wt{\omega}_t
\md v .
\end{align*}
The change of variables $w=\frac{\wt{\omega}_t}{(1+v)^{\alpha_s/2}}$ in the inner
integral gives
\begin{align*}
I_{a,1} &\leq
C_0C_\epsilon
\int_{0}^{\infty}
v^{(d-2)/2}(1+v)^{-\alpha_e}
\int_0^\infty
\frac{w^{2(a+\epsilon)}(1+v)^{(a+\epsilon-\alpha_t)\alpha_s}}{
(w^2 + 1)^{\alpha_t}
}
(1+v)^{\alpha_s/2}
\md w
\md v
\\&=
C_0C_\epsilon
\int_{0}^{\infty}
v^{(d-2)/2}(1+v)^{-\alpha_e-\alpha_s(\alpha_t-a-\epsilon-1/2)}
\int_0^\infty
\frac{w^{2(a+\epsilon)}}{
(w^2 + 1)^{\alpha_t}
}
\md w
\md v.
\end{align*}
In this expression, the inner integral is a finite constant, $C_w$, when $2\alpha_t-2a-2\epsilon > 1$, i.e., when $a+\epsilon < \alpha_t-1/2$. Since $\epsilon$ can be chosen arbitrarily small, we can make $C_w$ finite for all $a<\alpha_t-1/2$.
The remaining integral has an integrable singularity at $v=0$ for $d=1$, and the integral is finite when $\alpha_e + \alpha_s(\alpha_t-a-\epsilon-1/2)-(d-2)/2 > 1$. Solving for $a$ and again recognising that $\epsilon$ can be chosen arbitrarily small, we have now shown that $I_{a,1}<\infty$ when both $a<\alpha_t-1/2$ and $a<\frac{\alpha_e+(\alpha_t-1/2)-d/2}{\alpha_s}=\frac{\nu_s}{\alpha_s}$ hold. Therefore the temporal smoothness is
given by $\nu_t=\min(\alpha_t-1/2,\frac{\nu_s}{\alpha_s})$.

We now turn to the special case $d=2$, where we can derive an explicit
expression for the spectral density.
Let $B(x, y)$ be the beta function,
\begin{align*}
B(x,y) = \int_0^1 t^{x-1} (1-t)^{y-1} \md t.
\end{align*}
Making the change of variables $1+x=(1+v)^{\alpha_s}$ in \eqref{eq:initial-integration-steps}
the marginal temporal spectrum becomes
\begin{align*}
S_u(\omt)
&\propto \int_{0}^{\infty}
(1+x)^{-\frac{\alpha_e-1}{\alpha_s}-1}(\wt{\omega}_t^2 + 1 +x)^{-\alpha_t}\md x \qquad [\mbox{formula 3.197.9 in G\&R (p317)}]\\
&\propto B\left(\frac{\alpha_e-1}{\alpha_s}+\alpha_t, 1\right){}_{2}F_{1}\left(\alpha_t, \frac{\alpha_e - 1}{\alpha_s} + \alpha_t, \frac{\alpha_e - 1}{\alpha_s} + \alpha_t + 1; -\wt{\omega}_t^2\right),
\end{align*}
because $ \frac{\alpha_e - 1}{\alpha_s} + \alpha_t = \frac{\nu_s}{\alpha_s} + \frac{1}{2} > 0$.
Finally we verify that this spectrum yields the smoothness parameter implied by the general dimension result.
Assuming that $a-b$ is not an integer, the hypergeometric function ${}_{2}F_{1}(a,b;c,z)$ for large values values of $z$ behaves like
$$
{}_{2}F_{1}(a,b,c,z) \sim c_1 z^{-a} + c_2 z^{-b} + \mathcal{O}(z^{-a-1}) + \mathcal{O}(z^{-b-1})
$$
as $z \to \infty$.
If $a-b$ is an integer we have to multiply $z^{-a}$ or $z^{-b}$
with $\log(z)$ (\citet{erdelyi1953higher}  volume 1, section 2.3.2, page 76).
%(Reference: Erd\'elyi, Higher trancendental functions, ).
This extra logarithmic factor will not make a difference for the final smoothness. Thus, we may write
$$
S_t(\omt) = \mathcal{O}(\omt^{-2\alpha_t}) + \mathcal{O}\left(\omt^{-2(\frac{\alpha_e - 1}{\alpha_s} + \alpha_t)}\right) = \mathcal{O}\left(\omt^{-2(\alpha_t + \frac1{\alpha_s}\min(0,\alpha_e - 1))}\right)
$$
for large $\omt$.
This decay rate is such that the conditions in Theorem~\ref{thm:cramer} are if and only if $a<\nu_t$ with
$$
\nu_t = \frac{2(\alpha_t + \frac1{\alpha_s}\min(0,\alpha_e - 1)) - 1}{2} = \alpha_t + \frac1{\alpha_s}\min(0,\alpha_e - 1) - \frac{1}{2} =
\min\left[\alpha_t-\frac1{2}, \frac{\nu_s}{\alpha_s} \right],
$$
which completes the proof.

\subsection{Proof of Theorem~\ref{thm:spacetime-construction}}
\label{app:spacetime-construction}
Define the eigenvector matrix $\bm{V}$
and the eigenvalue (diagonal)
matrix $\bm{\Lambda}=\diag(\lambda_1,\dots,\lambda_{n_s})$ solving the
generalised eigenvalue
problem $\bm{K}_1 \bm{V} = \bm{C} \bm{V} \bm{\Lambda}$.
Since $\mv{K}_1$ and $\mv{C}$ are symmetric and $\mv{K}_1$ is positive definite,
the eigenvectors can be chosen so that
$\bm{V}^\top \bm{C} \bm{V} = \bm{I}$. For general $a=0,1,2,\dots$, $\mv{K}_{a+1}=\mv{K}_a\mv{C}^{-1}\mv{K}_1$, so that $\mv{K}_{a+1}\mv{V}=\mv{K}_a\mv{V}\mv{\Lambda}$. Recursion shows that $\mv{K}_a\mv{V}=\mv{C}\mv{V}\mv{\Lambda}^a$, which also holds for general $a\geq 0$, and $\mv{V}^\top\mv{K}_a\mv{V}=\mv{\Lambda}^a$.

For $\alpha_t=1$, the temporal evolution of the spatial Hilbert space discretisation of \eqref{eq:theorem-general-spde} is determined by
\begin{align*}
\left(\gamma_t\mv{C}\frac{\partial}{\partial t} + \mv{K}_{\alpha_s/2}\right) \mv{u}(t)
&=
\mv{C}
\dEE[\gamma_e^2 \mv{K}_{\alpha_e}](t), \quad t\in \RR .
\end{align*}
A multivariate change of variables $\mv{u}(t)=\mv{V}\mv{z}(t)$ and multiplication by $\mv{V}^\top$ on both sides gives
\begin{align*}
\left(\gamma_t\mv{I}\frac{\partial}{\partial t} + \mv{\Lambda}^{\alpha_s/2}\right) \mv{z}(t)
&=
\mv{V}^\top \mv{C} \dEE[\gamma_e^2 \mv{K}_{\alpha_e}](t) =
\dEE[\gamma_e^2\mv{\Lambda}^{\alpha_e}](t),
\end{align*}
where the precision of the driving noise process follows from
\begin{align*}
\gamma_e^2\left(\mv{V}^\top\mv{C}\mv{K}_{\alpha_e}^{-1}\mv{C}\mv{V}\right)^{-1}
&=
\gamma_e^2 \mv{V}^{-1}\mv{C}^{-1}\mv{K}_{\alpha_e}\mv{C}^{-1}\mv{V}^{-\top}
\\&=
\gamma_e^2 \mv{V}^{\top}\mv{K}_{\alpha_e}\mv{V}=\gamma_e^2 \mv{\Lambda}^{\alpha_e}
.
\end{align*}

For $\alpha_t=2$, the same technique yields
\begin{align*}
\left(-\gamma_t^2\mv{C}\frac{\partial^2}{\partial t^2} + \mv{K}_{\alpha_s}\right) \mv{u}(t)
&=
\mv{C}
\dEE[\gamma_e^2 \mv{K}_{\alpha_e}](t)
\end{align*}
and
\begin{align*}
\left(-\gamma_t^2\mv{I}\frac{\partial^2}{\partial t^2} + \mv{\Lambda}^{\alpha_s}\right) \mv{z}(t)
&=
\dEE[\gamma_e^2\mv{\Lambda}^{\alpha_e}](t) .
\end{align*}
Using the solutions for $\alpha_t=1$ and $2$ as the driving noise processes on the right hand side, the recursive construction technique from \citet{lindgren2011explicit} gives the general spatial discretisations
\begin{align*}
\left(
-\gamma_t^2\mv{C}\frac{\partial^2}{\partial t^2} + \mv{K}_{\alpha_s}
\right)^{\alpha_t/2}
\mv{u}(t)
&=
\mv{C}
\dEE[\gamma_e^2 \mv{K}_{\alpha_e}](t),\\
\left(
-\gamma_t^2\mv{I}\frac{\partial^2}{\partial t^2} + \mv{\Lambda}^{\alpha_s}
\right)^{\alpha_t/2}
\mv{z}(t)
&=
\dEE[\gamma_e^2\mv{\Lambda}^{\alpha_e}](t) ,
\end{align*}
for any $\alpha_t=1,2,\dots$.
Since the evolution of $\mv{z}(t)$ is independent between the vector components,
we get
\begin{align*}
\left(
-\gamma_t^2\frac{\partial^2}{\partial t^2} + \lambda_i^{\alpha_s}
\right)^{\alpha_t/2}
z_i(t)
&=
\frac{1}{\gamma_e\lambda_i^{\alpha_e/2}} \W_i(t) ,\quad \text{for $i=1,\dots,n_s$,}
\end{align*}
where $\lambda_i$ is the $i$:th generalised eigenvalue of $\mv{K}_1$, and $\W_i(\cdot)$ are
white noise processes, independent across all $i$. Rearranging factors, we get
\begin{align*}
\gamma_e\lambda_i^{\alpha_e/2}
\gamma_t^{\alpha_t}
\left(
-\frac{\partial^2}{\partial t^2} + \gamma_t^{-2}\lambda_i^{\alpha_s}
\right)^{\alpha_t/2}
z_i(t)
&=
\W_i(t) ,\quad \text{for $i=1,\dots,n_s$.}
\end{align*}
Applying the temporal condition of the theorem with $b_i=\gamma_e^2\lambda_i^{\alpha_e}\gamma_t^{2\alpha_t}$ and $\kappa_i=\lambda_i^{\alpha_s/2}/\gamma_t$ then gives a the temporal discretisation precision for each $z_i(t)$ as
\begin{align*}
\mv{Q}_{z_i} &=
\sum_{k=0}^{2\alpha_t} b_i \kappa_i^{2\alpha_t-k} \mb{J}_{\alpha_t,k/2}.
\end{align*}
Collecting the processes gives the joint precision as
\begin{align*}
\mv{Q}_{\mv{z}} &=
\sum_{k=0}^{2\alpha_t} \mb{J}_{\alpha_t,k/2}
\otimes
\diag(b_i \kappa_i^{2\alpha_t-k})
=
\gamma_e^2
\sum_{k=0}^{2\alpha_t}
\gamma_t^{k}
\mb{J}_{\alpha_t,k/2}
\otimes
\mv{\Lambda}^{\alpha_e+(2\alpha_t-k)\alpha_s/2}.
\end{align*}
The joint discretisation vector in the original parameterisation is given by
$\mv{u}=(\mv{I}\otimes\mv{V})\mv{z}$, with covariance $\mv{Q}_{\mv{u}}^{-1}=(\mv{I} \otimes \mv{V}) \mv{Q}_{\mv{z}}^{-1} (\mv{I} \otimes \mv{V}^{\top})$. We note that
$\mv{V}^{-\top}\mv{\Lambda}^a\mv{V}^{-1}=\mv{K}_a$,
so that the joint precision matrix becomes
\begin{align*}
\mv{Q}_{\mv{u}} &=
(\mv{I} \otimes \mv{V}^{-\top}) \mv{Q}_{\mv{z}} (\mv{I} \otimes \mv{V}^{-1})
=
\gamma_e^2
\sum_{k=0}^{2\alpha_t}
\gamma_t^{k}
\mb{J}_{\alpha_t,k/2}
\otimes
\mv{K}_{\alpha_e+(\alpha_t-k/2)\alpha_s},
\end{align*}
which completes the proof.

% !TEX root = spacetime_jrssb.tex
\section{Temporal GMRF representation with stationary boundary conditions}
\label{sec:temporal-disc}

We present precision matrices for
stationary AR(2) (autogregressive order 2) processes,
and then show how this can be used
to construct stationary boundary conditions for GMRF representations of 1st and
second order Whittle-Mat\'ern type stochastic differential equations.

\begin{lemma}\label{lemma:ar2boundary}
Let $u_k$ be a stationary AR(2) process with evolution
$$
a_0 u_k + a_1 u_{t-k} + a_2 u_{k-2} = e_k,
$$
with $a_0 > 0$ and $e_k$ independent, $e_k\sim N(0,1)$.
Then, the precision matrix $\bm{Q}$ for $(u_1,\dots,u_n)$ is quint-diagonal, and,
except for the upper left and lower right $2\times 2$ corners,
 $\bm{Q}$ has diagonal
elements elements $q_0 = a_0^2 + a_1^2 + a_2^2$ and off-diagonal elements $q_1 = a_1(a_0 + a_2)$ and $q_2 = a_0 a_2$.
Further, the corner elements are given by
\begin{align*}
&Q_{0,0} = Q_{n,n} = a_0^2,  &&
Q_{1,1} = Q_{n-1,n-1} = a_0^2 + a_1^2, \\
&Q_{0,1} = Q_{n,n-1} = a_1a_0, &&
Q_{1,0} = Q_{n-1,n} = a_1a_0.
\end{align*}

Conversely, if the inner elements $q_0$, $q_1$, and $q_2$ are known, the
$a_0$, $a_1$, and $a_2$ values can be recovered, and hence the corner elements
be constructed:
Define the constants
\begin{align*}
        b_+ = \sqrt{q_0+ 2q_1+2q_2}, \quad
        b_- = \sqrt{q_0-  2q_1+2q_2}, \quad
        b_s = \frac{b_+ + b_-}{2}.
\end{align*}
Then,
\begin{align*}
        a_0 = \frac{1}{2} \left(b_s + \sqrt{b_s^2 - 4 q_2} \right), \quad
        a_1 = \frac{b_+ - b_-}{2}, \quad
        a_2 = \frac{1}{2} \left(b_s - \sqrt{b_s^2 - 4 q_2} \right).
\end{align*}
\end{lemma}
\begin{proof}
        Follows by direct computation.
        \qqed
\end{proof}

Let $\Phi_t = \{\phi_1(t), ..., \phi_{N_t}(t) \}$ be a
set of piecewise linear basis functions in time,
on a regular grid,
and consider precision matrices on the
coefficients for a linear combination
of these basis functions.
We want to obtain a GMRF representation of a stationary process
Ornstein-Uhlenbeck process $z(t)$, such that
\begin{align}
\kappa z(t) +  \frac{d}{dt} z(t) = b^{-1/2} \epsilon(t),\quad t\in \RR \label{eq:ou1}
\end{align}
where $\epsilon$ is white noise.
However,
we can instead use the equivalent stochastic process model
\begin{align}
\left(\kappa^2  -  \frac{d^2}{dt^2} \right)^{1/2} z(t) = b^{-1/2} \epsilon(t),\quad t\in \RR.
\label{eq:ou2equiv}
\end{align}
Under stationarity, these two models
are equivalent in the sense that they have the same covariance function.
Let $\bm{M}_0 = \left(\langle \phi_i, \phi_j \rangle\right)_{i,j}$,
$\bm{M}_2 = \left(\langle \nabla \phi_i, \nabla \phi_j \rangle\right)_{i,j}$.
Assuming Neumann boundary conditions on a finite interval, and \eqref{eq:ou2equiv},
the precision matrix is
$\bm{Q} = b(\kappa^2 \bm{M}_0 + \bm{M}_2)$,
see  \citet[Sec 2.3]{lindgren2011explicit}.
This matrix does not represent a stationary process on the finite interval.
However, it is quint-diagonal,
and can be corrected to give
a stationary GMRF by adding
$b \kappa \sqrt{1 + h^{2} \kappa^{2}/4} \approx b \kappa$,
to the first
and the last entries of the matrix $\bm{Q}$,
per the previous lemma.
Here, $h$ is the step-size in the mesh, and we assume that $h \kappa$ is small.
Let $\bm M_1$ be a matrix of zeroes,
except the first and last elements which are $1/2$.
We then have a stationary GMRF representation
of the process \eqref{eq:ou2equiv}
with precision matrix
\begin{align}
b(\kappa^2 \bm{M}_0 + 2\kappa \bm{M}_1 + \bm{M}_2). \label{eq:gmrftime}
\end{align}

For second order B-spline basis functions, a similar adjustment can be made to the
initial and final 2-by-2 blocks of the matrix.  In both cases, Taylor expansion of the boundary
correction at a specific $\kappa_0>0$ is likely preferable when the temporal construction is applied to the space-time construction in Theorem~\ref{thm:spacetime-construction}.

\section{Application details}
\label{sec:app-details}

We performed the computations using nodes in the IBEX cluster at KAUST.
After preliminary model fitting with lower resolution spatial mesh we
fitted the model with $1\,251$ mesh nodes.
We used preliminary results to set initial values for the model parameters.
The computations were then carried out on a computer node with
a Intel (cascadelake) processor with $48$ threads and $3022.6$GB of RAM.
The parallel computations were performed with \texttt{inlabru} via \texttt{R-INLA}
with the \texttt{PARDISO} library, using
$3$ parallel evaluations of the posterior, each one using $16$ threads.
The average time per function evaluation were
$43.06$ seconds, $52.27$ seconds, $62.56$ seconds and $102.08$ seconds,
respectively for models $\Ma$, $\Md$, $\Mc$ and $\Md$.
The respective number of evaluations of the posterior density were
$243$, $259$, $211$ and $169$, and the total computing time
$3.17$ hours, $4.05$ hours, $3.96$ hours and $5.09$ hours.

The computed results were used for the within-sample and leave-one-out
prediction scores in Table~\ref{tab1gt}, as well as for the multi-horizon
forecast assessment in Section~\ref{sec:forecast-evaluation}.
Details of the multi-horizon
forecast scores are shown in Figures~\ref{fig:multi-horizon-months}
and~\ref{fig:multi-horizon-months-diff}, including
the mean error (ME, estimated forecast bias), mean absolute error (MAE),
mean squared error (MSE), mean Dawid-Sebastiani scores (DS), mean
continuous ranked probability score (CRPS), and scale-invariant CRPS (SCRPS).

\begin{figure}[t]
        \begin{center}
                \includegraphics[width=0.89\linewidth]{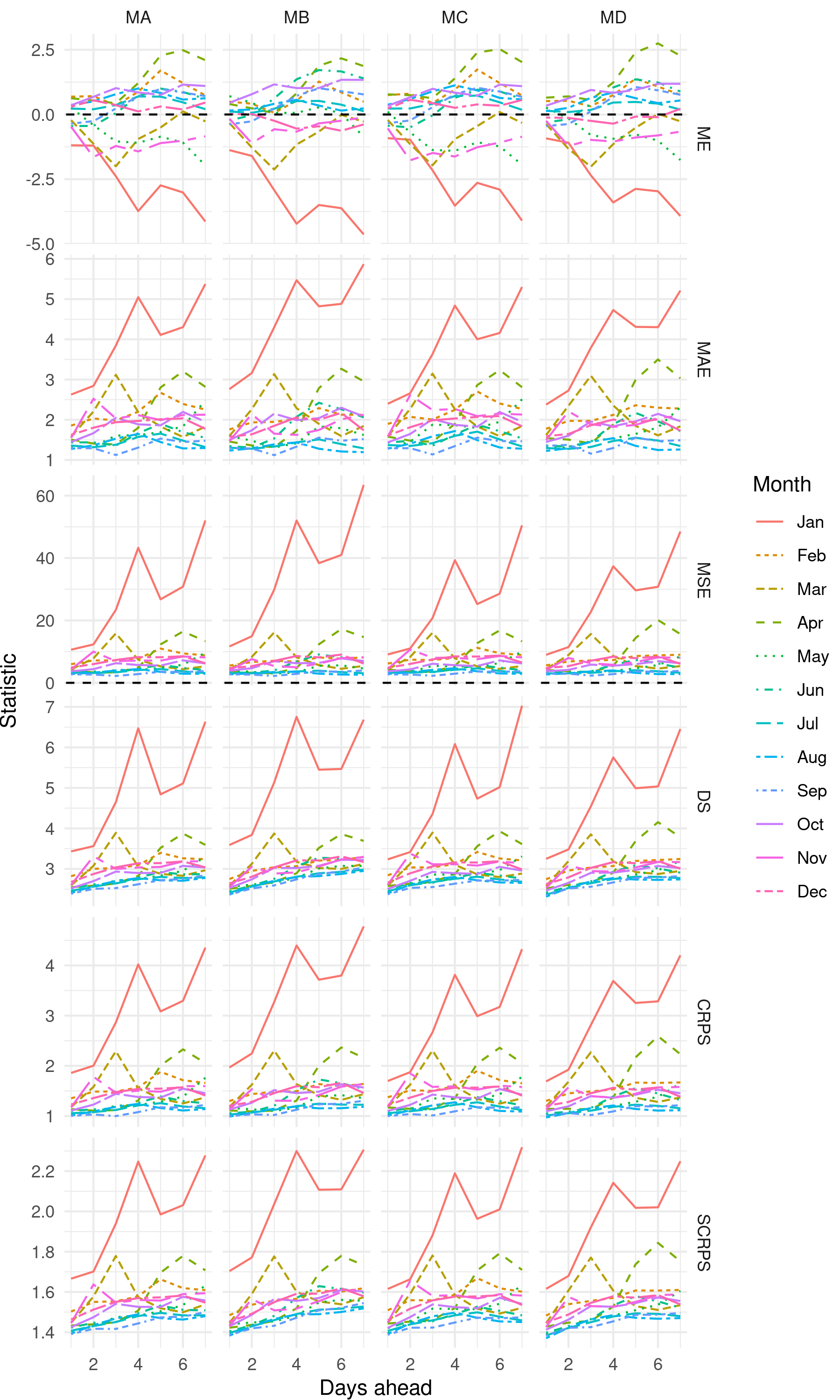}
        \end{center}
            \vspace{-0.5cm}
        \caption{
            Mean error and prediction score averages for each model, for each forecast horizon (1--7) and each month of the year, for the multi-horizon multi scenario setting.
            }
        \label{fig:multi-horizon-months}
\end{figure}
\begin{figure}[t]
        \begin{center}
                \includegraphics[width=0.89\linewidth]{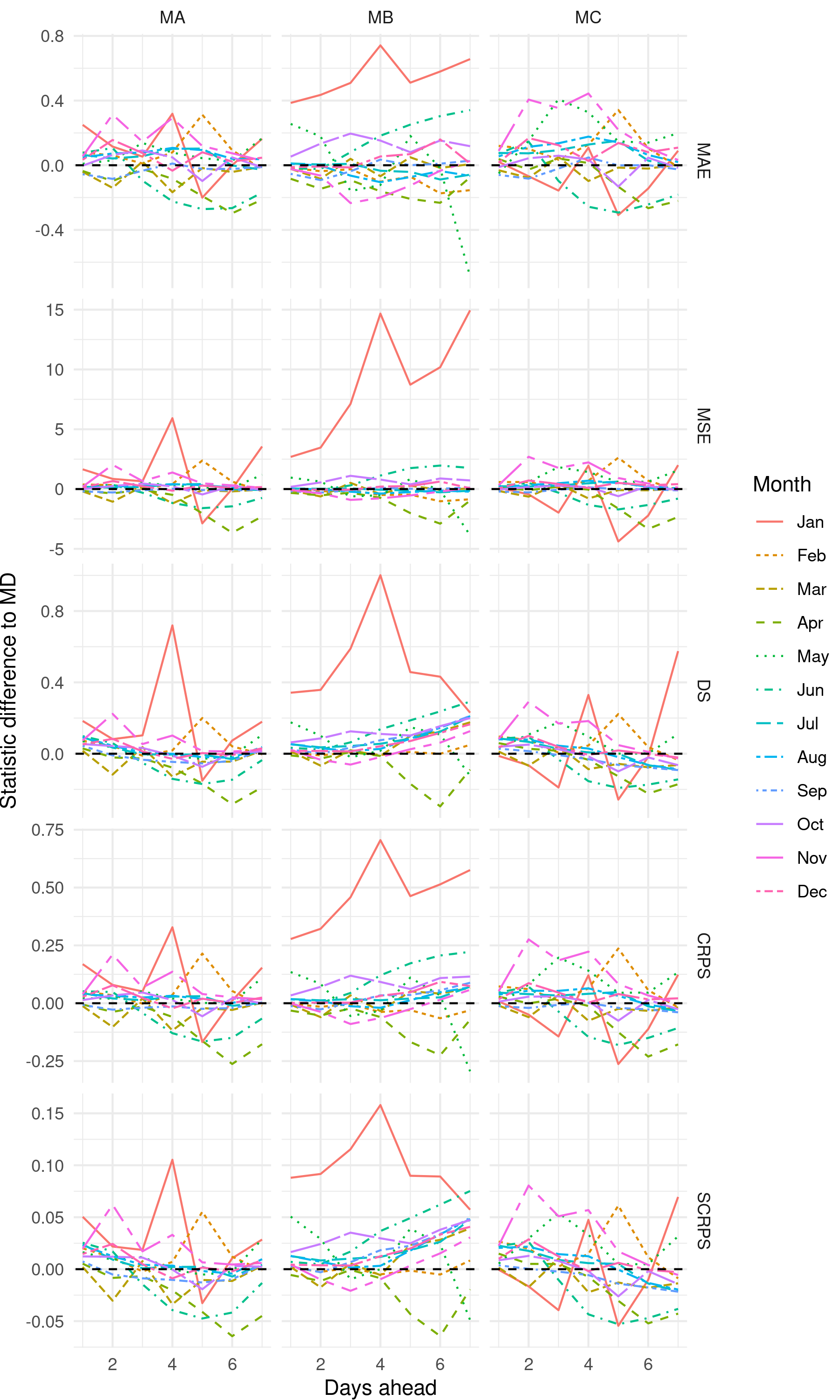}
        \end{center}
            \vspace{-0.5cm}
        \caption{
            Prediction score averages for each model with the scores for model $\Md$ subtracted,
            for each forecast horizon (1--7) and each month of the year, for the multi-horizon multi scenario setting.}
        \label{fig:multi-horizon-months-diff}
\end{figure}

\end{document}